\documentclass{article}
\usepackage{amsmath,amsfonts,amsthm,amssymb}
\usepackage{graphicx,psfrag,epsf}
\usepackage{enumerate}
\usepackage[numbers]{natbib}
\usepackage{color}
\usepackage{bm}
\usepackage[colorlinks,linkcolor=red,citecolor=blue,urlcolor=blue,filecolor=blue,backref=page,pdftex]{hyperref}
\usepackage{tikz}

\newtheorem{theorem}{Theorem}[section]

\newtheorem{lemma}[theorem]{Lemma}

\newtheorem{corollary}[theorem]{Corollary}
\newtheorem{example}{Example}
\theoremstyle{remark}
\newtheorem{remark}{Remark}

\DeclareMathOperator{\bmax}{\vee}         
\DeclareMathOperator{\bmin}{\wedge}       
\DeclareMathOperator{\argmax}{arg\,max}
\DeclareMathOperator*{\E}{\mathbb{E}}
\DeclareMathOperator{\betadist}{Beta}      
\DeclareMathOperator{\normaldist}{\mathcal{N}}      
\DeclareMathOperator{\range}{range}        

\DeclareBoldMathCommand{\ind}{1}

\newcommand{\Cplusplus}{C\nolinebreak\hspace{-.05em}\raisebox{.4ex}{\tiny\bf +}\nolinebreak\hspace{-.10em}\raisebox{.4ex}{\tiny\bf +}}

\renewcommand{\S}{\mathcal{S}}
\newcommand{\reals}{\mathbb{R}}
\newcommand{\posreals}{\reals^+}           
\newcommand{\A}{\mathcal{A}}               
\newcommand{\B}{\mathcal{B}}               
\newcommand{\intersection}{\cap}           
\newcommand{\comp}{\textnormal{c}}                     
\newcommand{\der}{\textnormal{d}}                      
\newcommand{\intder}{\,\der}                           
\newcommand{\thetamed}{\hat{\theta}^\textnormal{med}}  
\newcommand{\approxLambda}{\tilde{\Lambda}}             
\newcommand{\approxPi}{\tilde{\Pi}}                     
\newcommand{\approxq}{\tilde{q}}                        
\newcommand{\ceil}[1]{\lceil #1 \rceil}

\newcommand{\mlphi}{\hat{\beta}}                        
\newcommand{\mlalpha}{\hat{\alpha}}                    
\newcommand{\mlj}{j^*}    
\newcommand{\vphi}{\Phi}
\newcommand{\vpsi}{\Psi}

\title{Fast Exact Bayesian Inference for Sparse Signals
in the Normal Sequence Model}
\author{Tim van Erven\thanks{%
  Research supported by the Netherlands Organization for Scientific
  Research, grant number VI.Vidi.192.095}
  \ and 
  Botond Szabo\thanks{%
  Research supported by the Netherlands Organization for Scientific
  Research, grant number 016.Veni.173.040}\\
  Mathematical Institute\\
  Leiden University}

\begin{document}
\maketitle

\begin{abstract}
We consider exact algorithms for Bayesian inference with model selection
priors (including spike-and-slab priors) in the sparse normal sequence
model. Because the best existing exact algorithm becomes numerically
unstable for sample sizes over $n=500$, there has been much attention
for alternative approaches like approximate algorithms (Gibbs sampling,
variational Bayes, etc.), shrinkage priors (e.g.~the Horseshoe prior and
the Spike-and-Slab LASSO) or empirical
Bayesian methods.
However, by introducing algorithmic ideas from online sequential
prediction, we show that exact calculations are feasible for much larger
sample sizes: for general model selection priors we reach n=25\,000, and
for certain spike-and-slab priors we can easily reach n=100\,000. We
further prove a de Finetti-like result for finite sample sizes that
characterizes exactly which model selection priors can be expressed as
spike-and-slab priors. The computational speed and numerical accuracy of
the proposed methods are demonstrated in experiments on simulated data,
on a differential gene expression data set, and to compare the effect of
multiple hyper-parameter settings in the beta-binomial prior. In our
experimental evaluation we compute guaranteed bounds on the numerical
accuracy of all new algorithms, which shows that the proposed methods
are numerically reliable whereas an alternative based on long division
is not.
\end{abstract}

\noindent
\textbf{Keywords:} spike-and-slab prior, model selection,
high-dimensional statistics

\section{Introduction}

In the sparse normal sequence model we observe a sequence $Y=(Y_1,...,Y_n)$
that satisfies 
\begin{align}
Y_i=\theta_i+\varepsilon_i,\quad i=1,\ldots,n,\label{model: sequence}
\end{align}
for independent standard normal random variables $\varepsilon_i$, where
$\theta=(\theta_1,...,\theta_n)$ is the unknown signal of interest. It
is assumed that the number of non-zero signal components~$s$ in $\theta$
is small compared to the size of the whole sample (i.e.\ $s=o(n)$).
Applications of this model include detecting differentially expressed
genes
\cite{quackenbush:2002,Kramer:2013,Malone:2011,thomas:2001,efron:2012},
bankruptcy prediction for publicly traded companies using Altman's
Z-score in finance \cite{altman:1968,altman:2014}, separation of the
background and source in astronomical images
\cite{Clements:2012,Guglielmetti:2009}, and  wavelet analysis
\cite{abramovich2006,johnstone2004}. The model is further of interest to
sanity check (approximate) inference methods for the more general sparse
linear regression model (see \cite{CastilloSchmidtHieberVdVaart2015} and
references therein), which reduces to the normal sequence model when the
design is the identity matrix.

The sparse normal sequence model, which is also called the sparse normal
means model, has been extensively studied from a frequentist perspective
(see, for instance, \cite{Golubev2002,BirgeMassart2001,abramovich2006}),
but here we consider Bayesian approaches, which endow $\theta$ with a
prior distribution. This prior serves as a natural way to introduce
sparsity into the model and the corresponding posterior can be used for
model comparison and uncertainty quantification (see
\cite{fernandez:etal:2001,scott:2006,LiangEtAl2008,bayarri:etal:2012} and references therein). One
natural and well-understood class of priors are \emph{model selection
priors} that take the following hierarchical form:
\begin{enumerate}[i.)]
\item First a sparsity level $s$ is chosen from a prior $\pi_n$ on
$\{0,1,...,n\}$.
\item Then, given $s$, a subset of nonzero coordinates $\S \subset \{0,1,\ldots,n\}$
of size $|\S|=s$ is selected uniformly at random.
\item Finally, given $s$ and $\S$, the means $\theta_\S = (\theta_i)_{i \in \S}$ corresponding to the nonzero
coordinates in $\S$ are endowed with a prior $G_\S$, while the remaining
coefficients $\theta_{\S^\comp} = (\theta_i)_{i \notin \S}$ are set to zero.
\end{enumerate}
As is common, we will choose the prior $G_\S$ on the nonzero coordinates
in a factorized form; i.e.\ $\theta_i\sim G$ for all $i\in \S$, where
$G$ is a fixed one-dimensional prior, which we assume to have a
density~$g$ (with respect to the Lebesgue measure). Under suitable
conditions on $\pi_n$ and $G$, the posterior has good frequentist
properties and contracts around the true parameter at the minimax rate,
as shown by Castillo and Van der Vaart \cite{castillo2012}. Notably,
they require the prior $\pi_n$ to decrease at an exponential rate. 

A special case of the model selection priors are the spike-and-slab
priors developed by \citet{mitchell:1988,george:1997}, where the
coefficients of $\theta$ are assigned prior probabilities
\begin{equation}\label{prior:spike-and-slab}
  \begin{split}
\theta_i\mid \alpha&~\sim~ (1-\alpha)\delta_0+ \alpha G,\qquad i=1,...,n,\\
 \alpha&~\sim~ \Lambda_n,
  \end{split}
\end{equation}
with $\delta_0$ the Dirac-delta measure at $0$ (a \emph{spike}) and $G$
the same one-dimensional prior as above (called the \emph{slab} in this
context). The a priori likelihood of nonzero coefficients is controlled
by the \emph{mixing parameter} $\alpha \in [0,1]$, and finally
$\Lambda_n$ is a hyper-prior on~$\alpha$. A typical choice for
$\Lambda_n$ is the beta distribution: $\alpha\sim
\betadist(\kappa,\lambda)$. In this case the prior on
the sparsity level in the model selection formulation takes the form
$\pi_n(s) = \binom{n}{s} \frac{B(\kappa + s, \lambda + n -
s)}{B(\kappa,\lambda)}$, where $B(\kappa,\lambda)$ denotes the beta
function with parameters $\kappa$ and $\lambda$. The resulting prior is
called the \emph{beta-binomial prior}. A natural choice is $\kappa =
\lambda = 1$ \cite{scott:2010}, which corresponds to a uniform prior on
$\alpha$, but this choice does not satisfy the exponential decrease
condition on $\pi_n$. Castillo and Van der Vaart therefore propose
$\kappa=1$ and $\lambda=n+1$, which does satisfy their exponential
decrease condition \cite[Example~2.2]{castillo2012}, and in
Section~\ref{sec:asymptotics} we confirm empirically that the latter
indeed leads to better posterior estimates for $\theta$.

Model selection priors set certain signal components to zero, which is
desirable for model selection, but makes computation of the posterior
difficult since the number of possible sets $\S$ is exponentially large
(i.e.\ $2^n$). Castillo and Van der Vaart \cite{castillo2012} do provide
an exact algorithm, based on multiplication of polynomials (see
Appendix~\ref{sec:CvdV}), but this algorithm runs into numerical problems
for sample sizes over $n=250$ or sometimes $n=500$ (see
Section~\ref{sec:simulations}) and it also requires $O(n^3)$ computation
steps, which makes it too slow to handle large $n$.

The computational difficulty of model selection priors has given rise to
a variety of alternative priors based on shrinkage. These include the
horseshoe prior \cite{carvalho:2010}, for which multiple scalable
implementations are available \cite{pas:et:al:2016,johndrow:2017}. The
corresponding posterior achieves the minimax contraction rate and, under
mild conditions, also provides reliable uncertainty quantification
\cite{pas:etal:2014,pas:etal:2017,pas:etal:2017b}. The posterior median
and draws from the horseshoe posterior are not sparse, but one can use
it for model selection after post-processing the posterior. An
alternative is to replace the spike in the spike-and-slab prior with a
Laplace distribution with very small variance, as in the Spike-and-Slab
LASSO \cite{rovckova:2018}. One can efficiently compute the maximum a
posteriori (MAP) estimator of the corresponding posterior distribution by
convex optimization.

Another way to deal with the computational problems for model selection
priors is to consider approximations. The available options include
Stochastic Search Variable Selection (SSVS) \cite{george:1997},
variational Bayes approximation \cite{zhou:et:al:2009}, Langevin Markov
Chain Monte Carlo \cite{Pereyra:2016}, Expectation Maximization
\cite{george:rovckova:2014},  Hamiltonian Monte Carlo
\cite{scott:2006} or empirical Bayes methods \cite{johnstone2004,martin:2014,belitser:2020}.

In this paper we return to the goal of exactly computing the posterior
for model selection priors, without changing the prior or introducing
approximations. In Section~\ref{sec:hierarchical} we propose a new
approach based on a representation of model selection priors by a Hidden
Markov Model (HMM) that comes from the literature on online sequential
prediction and data compression \cite{VolfWillems1998}, for which we can
apply the standard Forward-Backward algorithm \cite{Rabiner1989}. The
computational complexity of this algorithm is $O(n^2)$. To appreciate
the speed-up compared to $O(n^3)$ run time, see
Section~\ref{sec:realworld}, where this method runs in under $15$
minutes while the previous algorithm of Castillo and Van der Vaart would
take approximately $20$ days. Furthermore, in
Section~\ref{sec:lambdabinomial} we specialize to spike-and-slab priors
and introduce an even faster algorithm based on a discretization of the
$\alpha$ hyper-parameter, which has only $O(n^{3/2})$ run time. Using
results from online sequential prediction \cite{DeRooijVanErven2009}, we
show that this discretization provides an accurate approximation of the
posterior that can be made exact to arbitrary precision, provided that
the density of $\Lambda_n$ varies sufficiently slowly. Our conditions do
not directly allow $\kappa$ or $\lambda$ to depend on $n$ in the
beta-binomial prior, so we provide an extra result to cover the
important case that $\kappa=1$ and $\lambda=n+1$. 

Our two new approaches allow us to easily handle data sets of size
$n=25\,000$ for general model selection priors and $n=100\,000$ for the
subclass of spike-and-slab priors with sufficiently regular $\Lambda_n$,
which both substantially exceed the earlier limit of $n=500$. These
results are obtained on a standard laptop within a maximum time limit of
half an hour. Run times for larger sample sizes can be estimated by
extrapolating from Figure~\ref{fig:compare:exact:alg}.

In Section~\ref{sec:connection} we further derive sufficient and
necessary conditions to decide whether a model selection prior can be
written in the more efficient spike-and-slab form. Since the
distribution of the binary indicators for whether $\theta_i = 0$ or not
is exchangeable under the model selection priors, this amounts to a
finite sample de Finetti result for a restricted class of exchangeable
distributions.

In Section~\ref{sec:simulations}, we demonstrate the scalability and
numerical accuracy of the proposed methods on simulated data.  We also
show there that our deterministic algorithm can be used as a benchmark
to test the accuracy of approximation methods: we compare the
approximate posterior from Gibbs sampling and variational Bayes to the
exact posterior computed by our algorithm, which shows the surprisingly
limited number of decimal places to which their answers are reliable.
Then, in Section~\ref{sec:realworld}, we compare our methods to other
approaches suggested in the literature in an application to differential
gene expression for Ulcerative Colitis and Crohn’s Disease. In
Section~\ref{sec:asymptotics} we further use our new algorithms to
empirically investigate the importance of the exponential decrease
condition on $\pi_n$ by varying the hyper-parameters $\kappa$ and
$\lambda$ of the beta-binomial prior. We find that exponential decrease
is not just a sufficient condition for minimax posterior contraction,
but it also leads to better posterior estimates of $\theta$. The paper
is concluded by Section~\ref{sec:discuss}, where we discuss possible
extensions of our algorithms.

In addition to the main paper, we provide an accompanying R package that
implements our new methods \cite{VanErvenDeRooijSzabo2019}, and
supplementary material with several appendices. In
Appendix~\ref{sec:CvdV} we first recall the exact algorithm by Castillo
and Van der Vaart \cite{castillo2012}. We show how to resolve its
numerical stability issues by performing all intermediate computations
in a logarithmic representation. The bottleneck then becomes its
computational complexity, because it requires $O(n^3)$ steps, which is
prohibitive for large~$n$. Two natural ideas to speed up the algorithm
have been proposed by \citep{castillo2012,Castillo2017}, one based on
fast polynomial multiplication and one based on long division.
Surprisingly, although both approaches look very promising in theory, it
turns out that neither of them works well in practice: the theoretical
speed-ups for fast polynomial multiplication turn out to be so
asymptotic that they do not provide significant gains for any reasonable
$n$; and the long division approach becomes numerically unstable again.
In Appendix~\ref{sec:simulation:extra} we provide an additional
variation on an experiment from Section~\ref{sec:asymptotics}. 
Finally, Appendix~\ref{sec:proofs} contains all proofs.

\section{Exact Algorithms for Model Selection Priors}
\label{sec:hierarchical}

In this section we propose novel, exact algorithms for computing
(marginal statistics of) the posterior distribution corresponding to
model selection priors. For general model
selection priors we propose a model selection HMM algorithm, and for
spike-and-slab priors we introduce a faster method based on
discretization of the $\alpha$ hyper-parameter. The section is concluded
with a characterization of the subclass of model selection priors that
can be expressed in the more efficient spike-and-slab form.

\paragraph{Marginal Statistics}
\label{sec:marginal}

We are interested in computing the marginal posterior probabilities
that the coordinates of $\theta$ are nonzero:
\[
  q_{n,i} := \Pi_n(\theta_i \neq 0 \mid Y)
  \quad \text{for $i = 1,\ldots,n$.}
\]
These are sufficient to compute any other marginal statistics of
interest, because, conditionally on whether $\theta_i$ is $0$ or not, the
pair $(Y_i,\theta_i)$ is independent of all other pairs $(Y_j,\theta_j)_{j\neq
i}$. For instance, the marginal posterior means can
be expressed as
\[
  \E[\theta_i \mid Y]
    = q_{n,i} \E[\theta_i \mid Y_i, \theta_i \neq 0]
    = q_{n,i} \,\frac{\zeta(Y_i)}{\psi(Y_i)},
\]
where $\psi(y) = \int \phi(y-t) g(t) \intder t$ is the slab density and
$\zeta(y) = \int t \phi(y-t) g(t) \intder t$, with $\phi$ the standard
normal density. We may also obtain marginal quantiles by inverting the
marginal posterior distribution functions
\[
  \Pi_n(\theta_i \leq u \mid Y)
    = (1-q_{n,i}) \ind_{u \geq 0}
      + q_{n,i} \frac{\psi(Y_i,u)}{\psi(Y_i)},
\]
where $\psi(y,u) = \int_{-\infty}^u \phi(y-t)g(t)\intder t$. In
particular, the marginal medians correspond to
\[
  \thetamed_i = \Big[H_{n,i}^{-1}\Big(\frac{1}{2 q_{n,i}}\Big) \bmin 0 \Big]
            + \Big[H_{n,i}^{-1}\Big(1-\frac{1}{2 q_{n,i}}\Big) \bmax 0 \Big],
\]
where $H_{n,i}^{-1}$ is the inverse of the function $H_{n,i}(u) =
\frac{\psi(Y_i,u)}{\psi(Y_i)}$ and we use the conventions that
$H_{n,i}^{-1}(v) = -\infty$ for $v \leq 0$ and $H_{n,i}^{-1}(v) = \infty$
for $v \geq 1$, see \citep{castillo2012}.

\subsection{The Model Selection HMM Algorithm}
\label{sec:HMM}

Our first computationally efficient approach is based on a Hidden Markov
Model (HMM) that comes from the literature on online sequential
prediction and data compression \citep{VolfWillems1998}. This approach
makes it possible to reliably compute all marginal posterior
probabilities $q_{n,i}$ in only $O(n^2)$ operations for any model
selection prior.

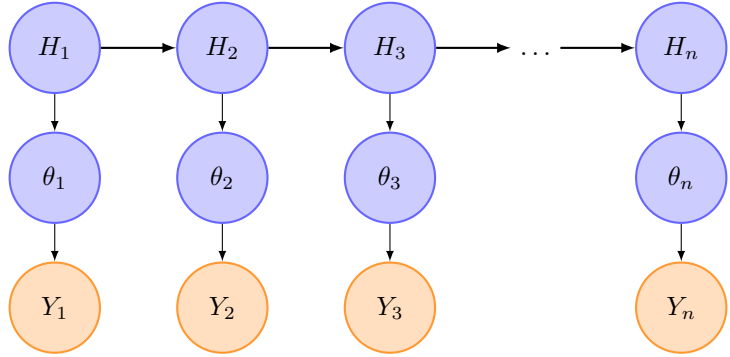
\begin{figure}
  \centering
  \tikzstyle{state}=[circle,thick,minimum size=1.2cm,draw=blue!60,fill=blue!20]
  \tikzstyle{observation}=[circle,thick,minimum size=1.2cm,draw=orange!80,fill=orange!25]

  \begin{tikzpicture}[>=latex,text height=1.5ex,text depth=0.25ex]
  \matrix[row sep=0.5cm,column sep=1cm]{%
    \node (H_1) [state] {$H_1$}; &
    \node (H_2) [state] {$H_2$}; &
    \node (H_3) [state] {$H_3$}; &
    \node (dots)        {$\ldots$}; &
    \node (H_n) [state] {$H_n$}; \\
    
    \node (theta_1) [state] {$\theta_1$}; & 
    \node (theta_2) [state] {$\theta_2$}; &
    \node (theta_3) [state] {$\theta_3$}; &
    &
    \node (theta_n) [state] {$\theta_n$}; \\

    \node (Y_1) [observation] {$Y_1$}; & 
    \node (Y_2) [observation] {$Y_2$}; &
    \node (Y_3) [observation] {$Y_3$}; &
    &
    \node (Y_n) [observation] {$Y_n$}; \\
  };

  \path[->]
    (H_1) edge[thick] (H_2)	
    (H_2) edge[thick] (H_3)	
    (H_3) edge[thick] (dots)	
    (dots) edge[thick] (H_n)	

    (H_1) edge (theta_1)	
    (H_2) edge (theta_2)	
    (H_3) edge (theta_3)	
    (H_n) edge (theta_n)	

    (theta_1) edge (Y_1)	
    (theta_2) edge (Y_2)	
    (theta_3) edge (Y_3)	
    (theta_n) edge (Y_n)	
  ;
  \end{tikzpicture}
  \caption{The model selection prior as a Hidden Markov Model}
  \label{fig:SSHMM}
\end{figure}

To define the HMM, we will encode the subset of nonzero coordinates $\S
\subset \{0,1,\ldots,n\}$ as a binary vector $B = (B_1,\ldots,B_n)$,
where $B_i = 1$ if $i \in \S$ and $B_i = 0$ otherwise. The crucial
observation is that the conditional probabilities of the model selection
prior 
\begin{equation}\label{eqn:n1sufficient}
  \Pi_n(B_{i+1} \mid B_1,\ldots,B_i)
    = \Pi_n(B_{i+1} \mid M_i)
\end{equation}
only depend on the total number of nonzeros $M_i = \sum_{j=1}^i B_j \in
\{0,\ldots,i\}$ in the first~$i$ coordinates and not on the locations of
these coordinates. We can use this observation to interpret the model
selection prior as the model selection HMM shown in
Figure~\ref{fig:SSHMM}, where each hidden state $H_i = (B_i, M_i)$
contains sufficient information to compute both the transition
probabilities
\[
  P(H_{i+1} \mid H_i) =
    \begin{cases}
      \Pi_n(B_{i+1} \mid M_i)
        & \text{if $M_{i+1} = M_i + B_{i+1}$,}\\
      0 & \text{otherwise,}
    \end{cases}
\]
and the conditional distribution of $\theta_i$ given $H_i$:
\begin{alignat*}{2}
  \theta_i &= 0 \quad (\text{a.s.}) &\qquad& \text{if $B_i = 0$,}\\
  \theta_i &\sim G && \text{if $B_i = 1$.}
\end{alignat*}
In fact, in our implementation we will integrate out $\theta_i$ to
directly obtain the conditional density
\begin{align*}
  p(Y_i \mid H_i) =
    \begin{cases}
      \phi(Y_i) & \text{if $B_i = 0$,}\\
      \psi(Y_i) & \text{if $B_i = 1$.}
    \end{cases} 
\end{align*}
(Note that $\psi(Y_i)$ is the conditional density of observation $Y_i$
for slabs, while $\phi(Y_i)$ is the density of $Y_i$ for spikes.)
Finally, the initial probabilities of $H_1$ are
\[
  P(H_1) =
    \begin{cases}
      \Pi_n(B_1) & \text{if $M_1 = B_1$,}\\
      0        & \text{otherwise.}
    \end{cases}
\]
We note that the sequence of hidden states $H_1,\ldots,H_n$ is in
one-to-one correspondence with $\S$. Consequently, since the model
selection HMM expresses the same joint distribution on $H_1,\ldots,H_n$
as the model selection prior, and the conditional distribution of
$\theta$ and $Y$ given $H_1,\ldots,H_n$ is also the same, it follows
that the model selection HMM is equivalent to the corresponding model
selection prior.

What we gain is that, for HMMs, standard efficient algorithms are
available, whose run times depend on the number of state transitions
with nonzero probabilities $P(H_{i+1} \mid H_i)$ \citep{Rabiner1989}.
For our purposes, we will use the Forward-Backward algorithm to compute
$\Pi_n(H_i \mid Y)$ for all~$i$ in $O(n^2)$ steps, from which we can
obtain $q_{n,i} = \Pi_n(B_i=1 \mid Y)$ for all~$i$ in another $O(n^2)$ steps by
marginalizing. For numerical accuracy, we perform all calculations using
the logarithmic representation discussed in
Appendix~\ref{sec:logarithmic}.

Let $Y_a^b = (Y_a,\ldots,Y_b)$. Then the Forward phase in this algorithm
computes the densities $p(Y_1^i, H_i = h_i)$ from $p(Y_1^{i-1}, H_{i-1} =
h_{i-1})$ for all $i=1,\ldots,n$ and all values $h_i$ of the hidden
states using the recursion
\begin{align*}
  p(Y_1^i, h_i) =
    \begin{cases}
      p(Y_1 \mid h_1) P(h_1)
        & \text{for $i=1$,}\\
      \displaystyle p(Y_i \mid h_i) \sum_{h_{i-1}} p(Y_1^{i-1},h_{i-1}) P(h_i
      \mid h_{i-1})
        & \text{for $1 < i \leq n$.}
    \end{cases}
\end{align*}
After the Forward phase, the
Forward-Backward algorithm performs the Backward phase, which computes
$p(Y_{i+1}^n \mid H_i = h_i)$ from $p(Y_{i+2}^n \mid H_{i+1} = h_{i+1})$
for all $i = n,\ldots,1$ using the recursion
\begin{align*}
  p(Y_{i+1}^n \mid h_i) =
    \begin{cases}
      1 & \text{for $i=n$,}\\
      \displaystyle \sum_{h_{i+1}} p(Y_{i+2}^n \mid h_{i+1}) p(Y_{i+1} \mid h_{i+1}) P(h_{i+1} \mid h_i)
        & \text{for $1 \leq i < n$.}
    \end{cases}
\end{align*}
Combining the results from the Forward and Backward phases, we can
compute
\[
  \Pi_n(h_i \mid Y) \propto p(Y_1^i, h_i) p(Y_{i+1}^n \mid h_i)
\]
for all $i$ and $h_i$ as desired.

The HMM described here was introduced by \citep{VolfWillems1998} for the
$\betadist(1/2,1/2)$-binomial prior (i.e.\ the spike-and-slab prior with
$\Lambda_n = \betadist(1/2,1/2)$) in the context of the Switching Method
for data compression. See \citep{KoolenDeRooij2013} for an overview of
many variations on this HMM. Indeed, for any
$\betadist(\kappa,\lambda)$-binomial prior this HMM is particularly
natural, because the transition probabilities of the hidden states have
a closed-form expression:
\[
  \Pi_n(B_{i+1}=1 \mid B_1,\ldots,B_i)
    = \Pi_n(B_{i+1}=1 \mid M_i)
    = \frac{\kappa + M_i}{\kappa + \lambda + i}.
\]
Here we add the observations that, even when the conditional
probabilities $\Pi_n(B_{i+1} \mid B_1,\ldots,B_i)$ are not available in
closed form for a given model selection prior, they still satisfy
\eqref{eqn:n1sufficient} and can be efficiently obtained from
\[
  \Pi_n(B_{i+1} \mid B_1,\ldots,B_i)
    = \frac{v_{i+1}(M_i + B_{i+1})}{v_i(M_i)}, 
\]
where $v_i(m) = \Pi_n(B_1 =
b_1,\ldots,B_i = b_i)$ is the joint probability of any sequence
$b_1,\ldots,b_i$ with $m$ ones. These joint probabilities can be
pre-computed for $i = n,\ldots,1$ in $O(n^2)$ steps using the recursion
\[
  v_i(m) = 
    \begin{cases}
      \pi_n(m)/\binom{n}{m}
        & \text{for $i = n$,}\\
      v_{i+1}(m) + v_{i+1}(m+1) 
        & \text{for $1 \leq i < n$.}
    \end{cases}
\]
Thus we can calculate the marginal posterior probabilities in $O(n^2)$
steps for any model selection prior, not just for beta-binomial priors.
The numerical accuracy of this algorithm is demonstrated in
Section~\ref{sec:simulations}.

\subsection{A Faster Algorithm for Spike-and-Slab Priors}
\label{sec:lambdabinomial}

In this section we restrict our attention to the spike-and-slab subclass
of model selection priors, for which we propose further speed-ups. It
is intuitively clear that the mixing hyper-parameter  $\alpha$ plays a
key role in the behavior of the prior distribution. The optimal choice
of $\alpha$ heavily depends on the sparsity parameter $s$ of the model.
For instance in case of Cauchy slabs the optimal oracle choice $\alpha=(s/n)\sqrt{\log(n/s)}$ results in minimax posterior contraction \cite{castillo2012} and reliable uncertainty quantification \cite{castillo:szabo:2018} in $\ell_2$-norm. However, in practice the
sparsity level $s$ is (typically) not known in advance. Therefore one
cannot use the optimal oracle choice for $\alpha$. In
\cite{castillo2012} it was also shown that by choosing $\alpha=1/n$ the
posterior contracts around the truth at the
nearly optimal rate $s\log(n)$. This seemingly solves the
problem of choosing the tuning hyper-parameter. However, a related simulation
study in \cite{castillo2012} shows that hard-thresholding at the corresponding $\sqrt{2\log (n)}$ level pairs up with substantially worse practical
performance; see Tables~1 and~2 in
\cite{castillo2012}.
Furthermore, in view of \cite{castillo:szabo:2018} the choice of
$\alpha=1/n$ imposes too strong prior assumptions, resulting in overly small
posterior spread which leads to unreliable Bayesian uncertainty
quantification, i.e.\ the frequentist coverage of the $\ell_2$-credible
set will tend to zero.

Therefore in practice one has to consider a data driven (adaptive)
choice of the hyper-parameter $\alpha$. A computationally appealing
approach is the empirical Bayes method, where the maximum marginal
likelihood estimator is plugged into the posterior. The corresponding
posterior mean achieves a (nearly) minimax convergence rate
\cite{johnstone2004}, and for slab distributions with polynomial tails
the corresponding posterior contracts around the truth at the optimal
rate \cite{castillo:mismer:2018}. However for light-tailed slabs (e.g.\
Laplace) the empirical Bayes posterior distribution will achieve a
highly suboptimal contraction rate around the truth; see again
\cite{castillo:mismer:2018}.

Another standard (and from a Bayesian perspective more natural) approach
is to endow the hyper-parameter $\alpha$ with another layer of prior
$\Lambda_n$. However, computational problems may arise using standard
Gibbs sampling techniques for sampling from the posterior; see Section
\ref{sec:compare:approx} for a demonstration of this problem on a
simulated data set. In the literature various speed-ups were proposed.
One can for instance focus on relevant sub-sequences of the sequential
parameter $\theta$ and apply the Gibbs sampler only on them. Another
approach is to apply the Hamiltonian Monte Carlo method, see for
instance \cite{scott:2006}. However, none of these approaches provides
an easy way to quantify their approximation error when run for a finite
number of iterations. In the next section we propose a
deterministic algorithm to approximate the marginal posterior
probabilities $q_{n,i}$ for spike-and-slab priors, with a guaranteed
bound on its approximation error that can be made arbitrarily close to
zero.

\subsubsection{Approximation via Discretization of the Mixing Parameter}\label{sec:discrete}

For general model selection priors the fast HMM algorithm from
Section~\ref{sec:HMM} requires $O(n^2)$ steps. However, for the special
case of spike-and-slab priors we can do even better: we can approximate
the corresponding posterior to arbitrary precision using only $O(n^{3/2})$ steps, provided that the density $\lambda_n$ of the mixing distribution $\Lambda_n$ on $\alpha$
satisfies certain regularity conditions.

\paragraph{The Algorithm}

Our approach is to approximate the prior $\Lambda_n$  by a prior $\approxLambda_n$
that is supported on $k = O(n^{1/2})$ discretization points
$\alpha_1,\ldots,\alpha_k$. Then let $\Pi_n$ be the original spike-and-slab prior corresponding to a
given choice of $\Lambda_n$, and let $\approxPi_n$ be the prior
corresponding to $\approxLambda_n$. Conditional on~$\alpha$, the pairs
$(\theta_i,Y_i)$ are independent. Computing the likelihood 
\[
\prod_{i=1}^n \Big((1- \alpha) \phi(Y_i) + \alpha \psi(Y_i)\Big)
\]
for a single $\alpha$ therefore takes $O(n)$ steps, and consequently we
can obtain the posterior probabilities $\approxPi_n(\alpha_j \mid Y)$ of
all $k$ discretization points in $O(kn)$ steps. We can then compute
\[
  \approxq_{n,i} := \approxPi_n(\theta_i \neq 0 \mid Y) = \sum_{j=1}^k
  \approxPi_n(\alpha_j \mid Y) \frac{\alpha_j \psi(Y_i)}{(1-\alpha_j)
  \phi(Y_i) + \alpha_j \psi(Y_i)}
\]
in another $O(k)$ steps independently for each $i$, leading to a total
run time of $O(kn) = O(n^{3/2})$ steps. We again perform all calculations
using the logarithmic representation from Appendix~\ref{sec:logarithmic}.

\paragraph{Choice of Discretization Points}

As in Section~\ref{sec:HMM}, let $B = (B_1,\ldots,B_n)$ be latent binary random
variables such that $B_i = 0$ if $\theta_i = 0$ and $B_i = 1$ otherwise.
We will choose discretization points $\alpha_1,\ldots,\alpha_k$ and the
discretized prior $\approxLambda_n$ such that the ratio
$\Pi_n(B=b)/\approxPi_n(B=b)$ is in $[1-\epsilon,1+\epsilon]$ for all
realizations $b$ of $B$, where $\epsilon > 0$ can be made arbitrarily
small. Since, conditional on $B$, the discretized model is the same as
the original model, this implies that the posterior probabilities
$\Pi_n(\theta \mid Y)$ and $\approxPi_n(\theta \mid Y)$ must also be
within a factor of $(1+\epsilon)/(1-\epsilon) \approx 1$.

Conditional on the mixing hyper-parameter $\alpha$, the sequence $B$ consists of
independent, identically distributed Bernoulli random variables, and
$\Pi_n$ and $\approxPi_n$ respectively assign hyperpriors $\Lambda_n$
and~$\approxLambda_n$ to the success probability $\alpha$. To discretize
$\alpha$, we will follow an approach introduced by
\cite{DeRooijVanErven2009} in the context of online sequential
prediction with adversarial data. They observe that it is more
convenient to reparametrize the Bernoulli model using the \emph{arcsine
transformation} \cite{Anscombe1948,FreemanTukey1950}, which makes the
Fisher information constant:
\begin{align*}
  \beta(\alpha) &= \arcsin \sqrt{\alpha}, &
  \alpha(\beta) &= \sin^2 \beta &
  \text{for $\beta \in [0,\pi/2]$.}
\end{align*}
We will use a uniform discretization of $\beta$ with $k$ discretization
points spaced $\delta_k = \pi/(2k) \propto 1/\sqrt{n}$ apart, which in the
$\alpha$-parametrization maps to a spacing that is proportional to
$1/\sqrt{n}$ around $\alpha=1/2$ but behaves like $1/n$ for $\alpha$
near~$0$ or~$1$. Specifically, let $\alpha_j = \alpha(\beta_j)$ with
\begin{align*}
  \beta_1 &= \frac{1}{2} \delta_k, &
  \beta_2 &= \frac{3}{2} \delta_k, &
  \beta_3 &= \frac{5}{2} \delta_k, &
  &\ldots, &
  \beta_k &= \frac{\pi}{2} - \frac{1}{2}\delta_k.
\end{align*}
The prior mass of each $\alpha$ under $\Lambda_n$ is then reassigned to
its closest discretization point in the $\beta$-parametrization. If
$\Lambda_n$ has no point-masses exactly half-way between discretization
points, then this means that
\begin{equation}\label{eqn:discreteprior}
  \approxLambda_n(\alpha_j)
    = \Lambda_n\big([\alpha(\beta_j - \delta_k/2), \alpha(\beta_j +
    \delta_k/2)]\big).
\end{equation}
Otherwise, if $\Lambda_n$ does have such point-masses, their masses may be
divided arbitrarily over their neighboring discretization points.

\paragraph{Approximation Guarantees}

For simplicity we will assume that $\Lambda_n$ has a Lebesgue-density
$\lambda_n(\alpha) = \der \Lambda_n(\alpha)/\der \alpha$. It will also
be convenient to let $\alpha_0 = 0$ and $\alpha_{k+1} = 1$, and to
define
\[
  P_\alpha(n,\mlalpha) = \alpha^{\mlalpha n}(1-\alpha)^{(1-\mlalpha)n},
    \qquad \text{for $\mlalpha \in [0,1]$, $n \in \posreals$,}
\]
which may be interpreted as the Bernoulli$(\alpha)$ likelihood of a
binary sequence with maximum likelihood parameter $\mlalpha$. In
particular, if $\mlalpha = s/n$ with $s$ the number of ones in $b \in
\{0,1\}^n$ for integer $n$, then 
\begin{align*}
  \Pi_n(B=b) &= \int_0^1 P_\alpha(n,\mlalpha) \lambda_n(\alpha) \der
      \alpha,
      &
  \approxPi_n(B=b)
    &= \sum_{j=1}^k P_{\alpha_j}(n,\mlalpha) \approxLambda_n(\alpha_j).
\end{align*}
There is no reason to restrict the definition of $P_\alpha(n,\mlalpha)$
to integer $n$ or to the discrete set of $\mlalpha$ that can be maximum
likelihood parameters at sample size $n$, however, and following
\citep{DeRooijVanErven2009} we extend the definition to all $\mlalpha
\in [0,1]$ and all real $n > 0$, which will be useful below to handle
the $\betadist(1,n+1)$ prior.
\begin{theorem}\label{thm:discrete}
  Take $k = 2 (m+1) \ceil{\sqrt{n}\,} + 1$ for any integer $m$, and
  suppose there exists a constant $L\geq 0$ (which is allowed to depend on
  $n$) such that
  \begin{equation}\label{cond:prior}
    \frac{\sup_{\alpha \in [\alpha_j,\alpha_{j+1}]}
    \lambda_n(\alpha)\sqrt{\alpha(1-\alpha)}}{\inf_{\alpha \in[\alpha_j,\alpha_{j+1}]}
    \lambda_n(\alpha)\sqrt{\alpha(1-\alpha)}}\leq e^{L\sqrt{n}\delta_k},
    \quad\text{for all $j=0,...,k$.}
  \end{equation}
  Then there exists a constant $C_L > 0$ that depends only on $L$, such
  that, if $m > C_L$, we have, for $\epsilon = C_L/m$,
  \begin{equation}\label{eqn:approxBernoulli}
    (1-\epsilon) 
      \leq \frac{\int_0^1 P_\alpha(n,\mlalpha) \lambda_n(\alpha) \der
      \alpha}{\sum_{j=1}^k P_{\alpha_j}(n,\mlalpha) \approxLambda_n(\alpha_j)}
      \leq (1+\epsilon),
    \qquad \text{for all $\mlalpha \in [0,1]$,}
  \end{equation}
  and consequently
  \begin{equation}\label{eqn:posteriorratio}
    \frac{1-\epsilon}{1+\epsilon}
      \leq \frac{\Pi_n(\theta \mid Y)}{\approxPi_n(\theta \mid Y)}
      \leq \frac{1+\epsilon}{1-\epsilon}
    \qquad \text{almost surely.}
  \end{equation}
\end{theorem}
The result \eqref{eqn:approxBernoulli} holds even for non-integer $n$,
but \eqref{eqn:posteriorratio} implicitly assumes that $n$ is the number
of observations in $Y$ and must therefore be integer. The proof is
deferred to Appendix~\ref{sec:proofdiscrete} in the supplementary
material. We note already that condition~\eqref{cond:prior} is
essentially a Lipschitz condition on the log of the density of
$\Lambda_n$ in the $\beta$-parametrization (see \eqref{cond:phiprior} in
Appendix~\ref{sec:proofdiscrete}). Under this condition, the theorem
shows that, by increasing~$m$, we can approximate $\Pi_n(\theta \mid Y)$
to any desired accuracy, at the cost of increasing our computation time,
which scales linearly with $m$.

\begin{remark}
For given $m$ and (integer) $n$, the
tightest possible value of $\epsilon$ in \eqref{eqn:posteriorratio} may
be determined numerically, by maximizing and minimizing the ratio in
\eqref{eqn:approxBernoulli} over $\mlalpha = s/n$ for $s=0,\ldots,n$.
\end{remark}

\paragraph{Extension to Arbitrary Beta Priors}

The Lipschitz condition \eqref{cond:prior} excludes the important
$\betadist(1,n+1)$ prior, because its density varies too rapidly. We
therefore describe an extension that can handle any
$\betadist(\kappa,\lambda)$ prior with $\kappa,\lambda \geq 1/2$, even
when $\kappa$ or $\lambda$ grows linearly with $n$.

To this end, we interpret $\betadist(\kappa,\lambda)$ as the posterior
of a $\betadist(1/2,1/2)$ prior after observing $\kappa-1/2$ fake ones
and $\lambda-1/2$ fake zeros. Our effective sample size for the fake
observations and the real data together is then $n' = n + \kappa +
\lambda - 1$ (which need not be an integer). Since $\betadist(1/2,1/2)$
is uniform in the $\beta$-parametrization, it satisfies
\eqref{cond:prior} with the best possible constant: $L=0$, so applying
Theorem~\ref{thm:discrete} we find that \eqref{eqn:approxBernoulli}
holds for sample size $n'$ with $\epsilon$ as in the theorem. 
We then take the discretization points for sample size $n'$ with
corresponding discrete prior $\approxLambda_{n'}$ defined by
\eqref{eqn:discreteprior} (which is actually uniform, with probabilities
$1/k$, because $\betadist(1/2,1/2)$ is uniform in the
$\beta$-parametrization), and we compute a new prior $\approxLambda_n$ on
these discretization points as the posterior from $\approxLambda_{n'}$
after observing $\kappa -1/2$ fake ones and $\lambda-1/2$ fake zeros:
\begin{equation}
  \label{eqn:fastforwardbeta}
  \approxLambda_n(\alpha_j) = 
    \frac{\frac{1}{k}\, \alpha_j^{\kappa-1/2}(1-\alpha_j)^{\lambda-1/2}}
         {\sum_{j'=1}^k \frac{1}{k}\,
         \alpha_{j'}^{\kappa-1/2}(1-\alpha_{j'})^{\lambda-1/2}}
  \qquad \text{for $j=1,\ldots,k$.}
\end{equation}
\begin{corollary}\label{cor:betaprior}
  For any $\kappa \geq 1/2,\lambda \geq 1/2$ and positive integer $n$,
  let $k = 2 (m+1) \ceil{\sqrt{n'}\,} + 1$, where $n' = n + \kappa +
  \lambda - 1$ and $m > C_0$ is any integer that exceeds the constant
  $C_L$ from Theorem~\ref{thm:discrete} for $L=0$. Let $\Lambda_n$ be
  the $\betadist(\kappa,\lambda)$ prior and let $\approxLambda_n$ be as
  in \eqref{eqn:fastforwardbeta}. Then \eqref{eqn:approxBernoulli} and
  \eqref{eqn:posteriorratio} hold with $\epsilon' =
  2\epsilon/(1-\epsilon)$ instead of $\epsilon = C_0/m$.
\end{corollary}

\begin{proof}
  Since the joint distributions on $n'$ observations satisfy
  \eqref{eqn:approxBernoulli}, the corresponding posteriors after
  conditioning these distributions on $\kappa-1/2$ fake ones and
  $\lambda-1/2$ fake zeros must be within factors
  $\frac{1-\epsilon}{1+\epsilon} = 1-\frac{2\epsilon}{1+\epsilon} \geq
  1-\epsilon'$ and $\frac{1+\epsilon}{1-\epsilon} = 1+\epsilon'$.
\end{proof}

\subsection{Which Model Selection Priors Are Spike-and-Slab Priors?}\label{sec:connection}

As described in the introduction, it is clear that spike-and-slab priors
are a special case of model selection priors. However, to the best of
our knowledge, it is not known when a model selection prior has a
spike-and-slab representation. One advantage of the spike-and-slab
formulation is that we can construct algorithms with $O(n^{3/2})$
run time (Section~\ref{sec:lambdabinomial}), while for a general model
selection prior the computational complexity is $O(n^2)$
(Section~\ref{sec:HMM}). In this section we give sufficient and
necessary conditions for when a model selection prior can be expressed
in spike-and-slab form.

To characterize the exact relationship between the
priors we introduce the following notation. For
$\mu=(\mu_0,\mu_1,...,\mu_{m})$ with $m\geq 2n$, define the $(n+1)\times
(n+1)$ Hankel matrix $H_n(\mu)=[\mu_{i+j}]_{i,j=0,...,n}$ and let
$F\mu=(\mu_1,...,\mu_{m})$ denote the projection that drops the first
coordinate. Furthermore, for $A\in \reals^{n\times m}$, let $\range(A)$
be the column space of $A$ and let $A\succeq 0$ denote that $A$ is
positive semi-definite.
\begin{theorem}\label{thm:equivalence:prior}
For odd $n=2k+1$, the model selection prior  (with factorizing $G_\S$) can be given in
the form \eqref{prior:spike-and-slab}  if and only if there exists a $c_n\in[0,\pi_n(n)]$ such that
\begin{align*}
H_k(\mu)\succeq 0, \quad H_k(F\mu)\succeq 0,\quad \text{and}\quad
(\mu_{k+1},\mu_{k+2},...,\mu_{2k+1})^\top \in \range\big(H_k(\mu)\big),
\end{align*}
with $\mu=\Big({n \choose 0}^{-1}\pi_n(0),\ldots, {n \choose
n-1}^{-1}\pi_n(n-1), c_n
\Big)\in[0,1]^{n+1}$.

For even $n=2k$, the model selection prior  (with factorizing $G_\S$) can be given in
the form \eqref{prior:spike-and-slab}  if and only if there exists a $c_n\in[0,\pi_n(n)]$ such that
\begin{align*}
H_k(\mu)\succeq 0, \quad H_{k-1}(F\mu)\succeq 0,\quad \text{and}\quad
(\mu_{k+1},\mu_{k+2},...,\mu_{2k})^\top \in \range\big(H_{k-1}(F\mu)\big),
\end{align*}
with the same $\mu$ as above.
 \end{theorem}

The proof, which is given in Appendix~\ref{sec:equivalence:prior}, shows
that establishing this theorem amounts to proving a version of de
Finetti's theorem for finite sequences.

Next we give several examples of priors $\pi_n$ that satisfy (or fail)
the conditions of Theorem~\ref{thm:equivalence:prior}, which implies
that the model selection prior can (or cannot) be given in
spike-and-slab form \eqref{prior:spike-and-slab}. The proofs for the
examples are in Appendix~\ref{sec:examples}.

First we consider binomial $\pi_n$, for which it is already known that
there exists a spike-and-slab representation
\cite[Example~2.1]{castillo2012}. Nevertheless, to illustrate the
applicability of our results, we show that this choice of $\pi_n$
satisfies the conditions of Theorem~\ref{thm:equivalence:prior}.
\begin{example}\label{example:1}
The binomial prior $\pi_n(s)\propto {n \choose s} p^s(1-p)^{n-s}$, $p\in[0,1]$,
satisfies the conditions of Theorem \ref{thm:equivalence:prior} and
therefore the corresponding model selection prior can be given in the
spike-and-slab form \eqref{prior:spike-and-slab} for some appropriate
probability measure $\Lambda_n$ on $[0,1]$.
\end{example}

The next example treats the Poisson prior as a choice for $\pi_n$. To
the best of our knowledge there are no results in the literature that
establish whether the corresponding model selection prior can be given in the spike-and-slab form \eqref{prior:spike-and-slab}.
\begin{example}\label{example:2}
For any $\lambda > 0$, the Poisson prior
$\pi_n(s)\propto e^{-\lambda} \lambda^s/s!$ restricted to $s \in
\{0,1,...,n\}$ satisfies the conditions of Theorem
\ref{thm:equivalence:prior} and therefore the corresponding model
selection prior can be given in the form
\eqref{prior:spike-and-slab} for some appropriate probability measure
$\Lambda_n$ on $[0,1]$.
\end{example}

We proceed to give two natural choices for $\pi_n$ where the
corresponding model selection prior cannot be expressed in the form
\eqref{prior:spike-and-slab}. In the first example, $\pi_n$ has a heavy
(polynomial) tail, while in the second it has a light (sub-exponential)
tail.

\begin{example}\label{example:3}
Let us consider the prior $\pi_n(0)\propto1$, $\pi_n(s)\propto s^{-\lambda}$, $s=1,...,n$, 
for any $\lambda>1$. For $n> 2^{\lambda-1}/(2^{\lambda-1}-1)$ this prior does not satisfy the conditions of Theorem~\ref{thm:equivalence:prior} and therefore the corresponding model selection prior cannot be represented in the form~\eqref{prior:spike-and-slab}.
\end{example}

\begin{example}\label{example:4}
We consider the sub-exponential prior $\pi_n(s)\propto e^{-s^\lambda}$, $s=0,1,...,n$
for any $\lambda>\log_2 (2+\ln 2)$. For $n>c/(c-1)$ with
$c=e^{2^{\lambda}-2}/2>1$ this prior does not satisfy the conditions of
Theorem~\ref{thm:equivalence:prior} and therefore the corresponding model
selection prior cannot be represented in the form~\eqref{prior:spike-and-slab}.
\end{example}

\section{Simulation Study: Reliability of Algorithms}
\label{sec:simulations}

\subsection{Comparing the Proposed Algorithms}
\label{sec:compare:exact}

\begin{figure}[tbp]
  \centering
\begin{minipage}[b]{0.48\textwidth}
    \includegraphics[width=\textwidth]{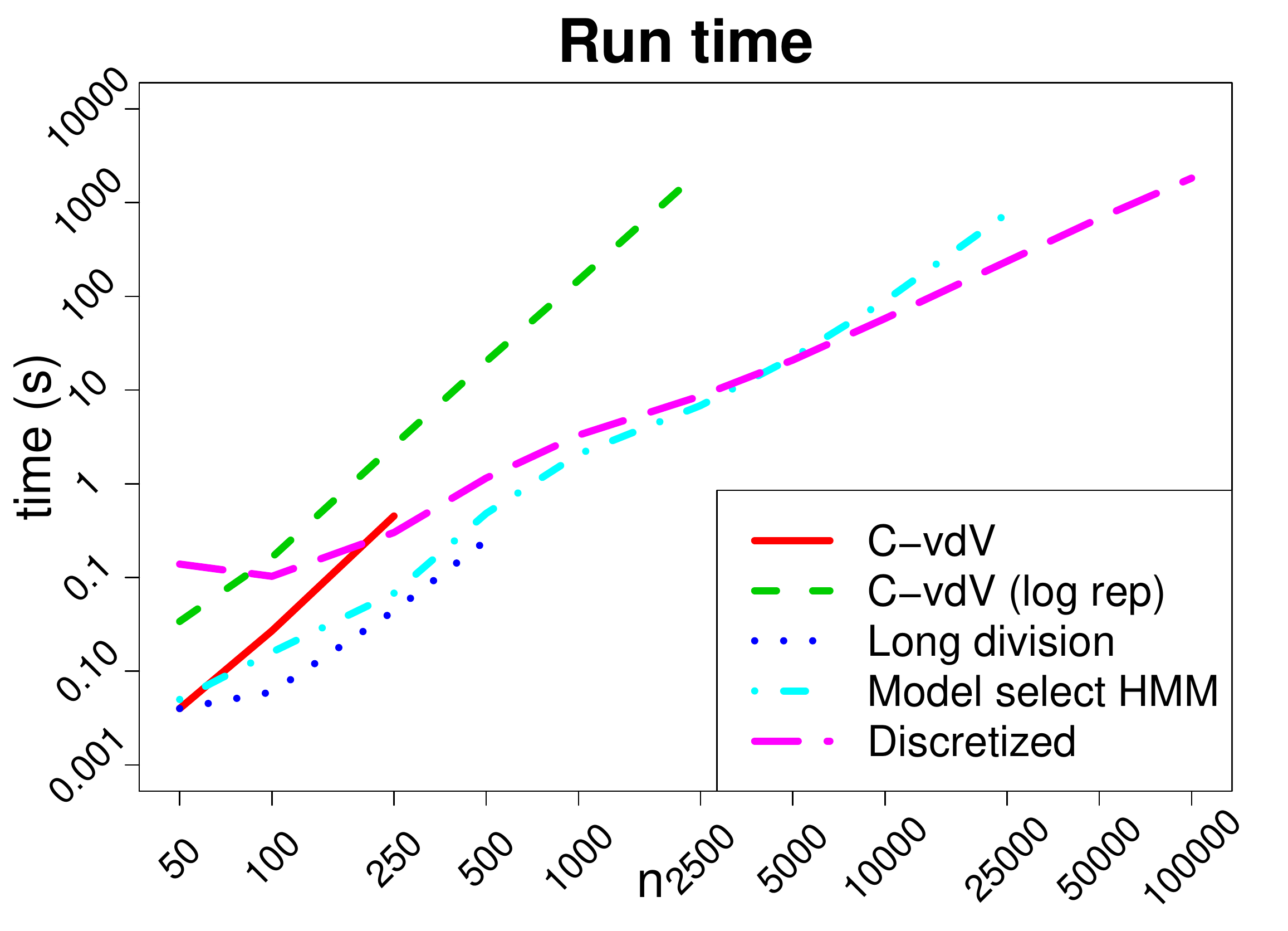}
  \end{minipage}
  \hfill
  \begin{minipage}[b]{0.48\textwidth}
    \includegraphics[width=\textwidth]{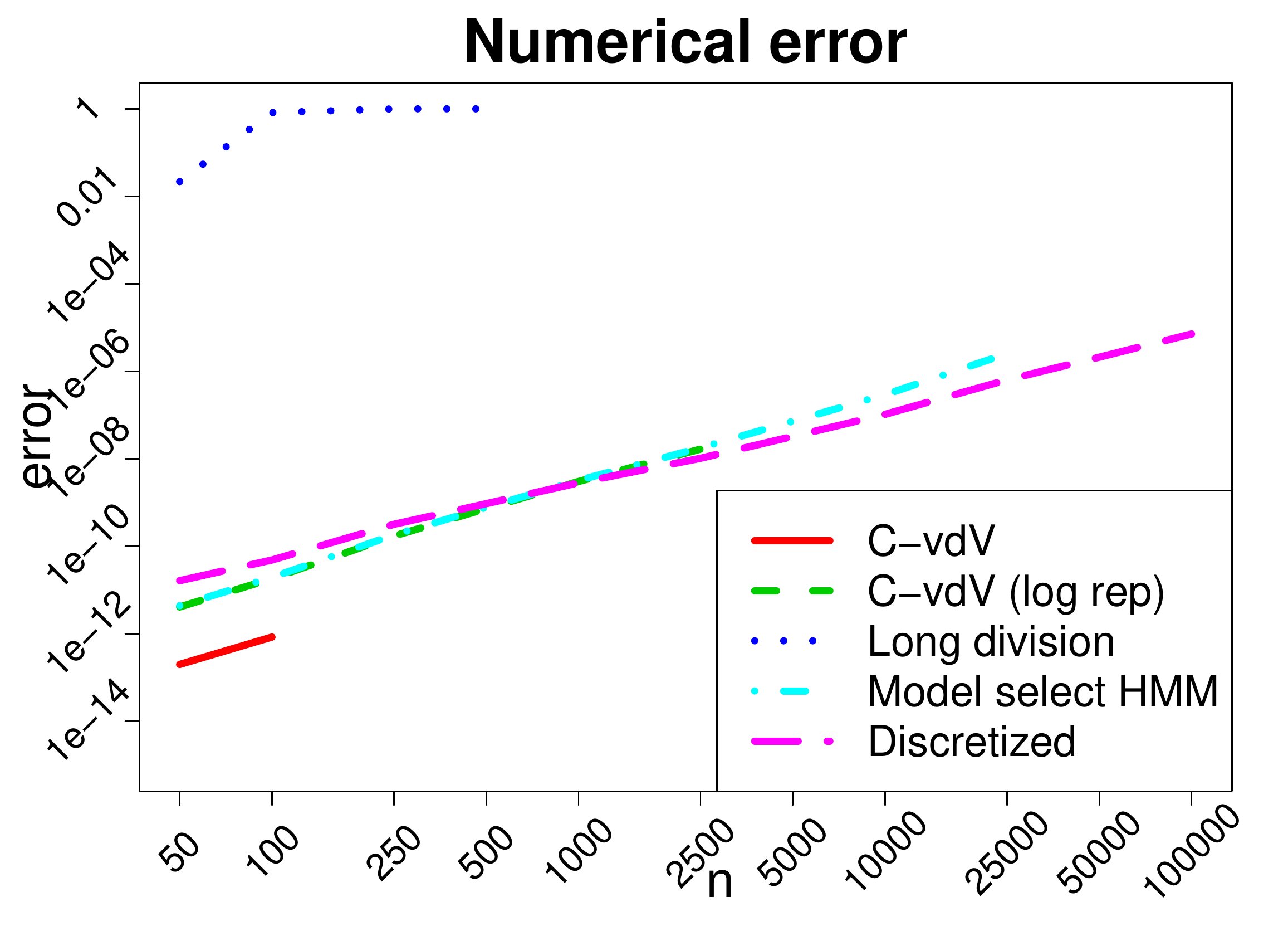}
  \end{minipage}
\caption{Run time and numerical accuracy for the exact algorithms from
Section~\ref{sec:hierarchical} in calculating $q_{n,i}$ for
$i=1,\ldots,n$. A numerical error of $10^{-a}$ means that the algorithm
is able to calculate the mathematically exact answer up to $a$ decimal
places.}
\label{fig:compare:exact:alg}
\end{figure}

In this section we investigate the speed and numerical accuracy of our
new algorithms to the previously proposed methods for exact computation
of the posterior.
We consider a sequence of sample sizes
$n=50,\allowbreak 100,\allowbreak 250,\allowbreak 500,\allowbreak
1\,000,\allowbreak 2\,500,\allowbreak ...,\allowbreak
50\,000,\allowbreak 100\,000$ and construct the true signal
$\theta_0$ to have $20\%$ non-zero signal components of value
$4\sqrt{2\ln n}$, while the rest of the signal coefficients are set to
be zero. For fair comparison we run all algorithms for the
spike-and-slab prior with Laplace slab $g(x) = \frac{a}{2}e^{-a|x|}$, with $a=1$, and
mixing hyper-prior $\Lambda_n=\betadist(1,n+1)$.  We have set up the
experiments in R, but all algorithms were implemented as subroutines in
\Cplusplus{}. Since numerical instability is a major concern, we have tracked the
numerical accuracy of all methods using interval arithmetic as
implemented in the \Cplusplus{} Boost library \cite{boost} (with cr-libm as a
back-end to compute transcendental functions \cite{CRLibm}), which
replaces all floating point numbers by intervals that are guaranteed to
contain the mathematically exact answer. The lower end-point of each
interval corresponds to always rounding down in the calculations, and
the upper end-point corresponds to always rounding up. The width of the
interval for the final answer therefore measures the numerical error.
All experiments were performed on a MacBook Pro laptop with 2.9 GHz
Intel Core i5 processor, 8~GB (1867 MHz DDR3) memory, and a solid-state
hard drive.

\paragraph{Results}

The results are summarized in Figure~\ref{fig:compare:exact:alg}, which shows the run time of the algorithms
on the left, and their numerical error on the right. The reported
numerical error is the maximum numerical error in calculating $q_{n,i}$
over $i = 1,\ldots,n$. To avoid overly long computations we have
terminated the algorithms if they became numerically unstable or if
their run time exceeded half an hour. One can see that the original
Castillo-Van der Vaart algorithm was terminated for $n\geq 250$, which
was due to numerical inaccuracy. This problem was resolved by applying
the logarithmic representation from Appendix~\ref{sec:logarithmic} which
made the algorithm numerically stable up to $n\leq 2500$; however, due
to the long $O(n^3)$ run time the algorithm was terminated for larger
values as it reached the half-hour limit. The natural speed-up idea of
applying long division (see Appendix~\ref{sec:speedupCVdV}) was not
successful for this data as even for small sample sizes the numerical
accuracy was poor. We observe that the model selection HMM and the
algorithm based on discretization performed superior to the preceding
methods: the model selection HMM algorithm has run time $O(n^2)$ and the
largest sample size it managed to complete within half an hour was
$n=25\,000$, while the algorithm with discretized mixing parameter in
the spike-and-slab prior (initialized according to
Corollary~\ref{cor:betaprior} with parameter $m=20$) has run time
$O(n^{3/2})$ and reached the time limit after sample size $n=100\,000$.
We also note that both algorithms were numerically accurate, giving
answers that were reliable up to between 5 and 11 decimal places,
depending on $n$. For sample size $n=2\,500$, we have further verified
empirically that indeed the model selection HMM algorithm computes the
same numbers as the Castillo-Van der Vaart algorithm, as was already
shown in Section~\ref{sec:HMM}.

\subsection{Approximation Errors for Several Standard Methods}\label{sec:compare:approx}

\begin{figure}[tbp]
  \centering
  \begin{minipage}[b]{0.48\textwidth}
    \includegraphics[width=\textwidth]{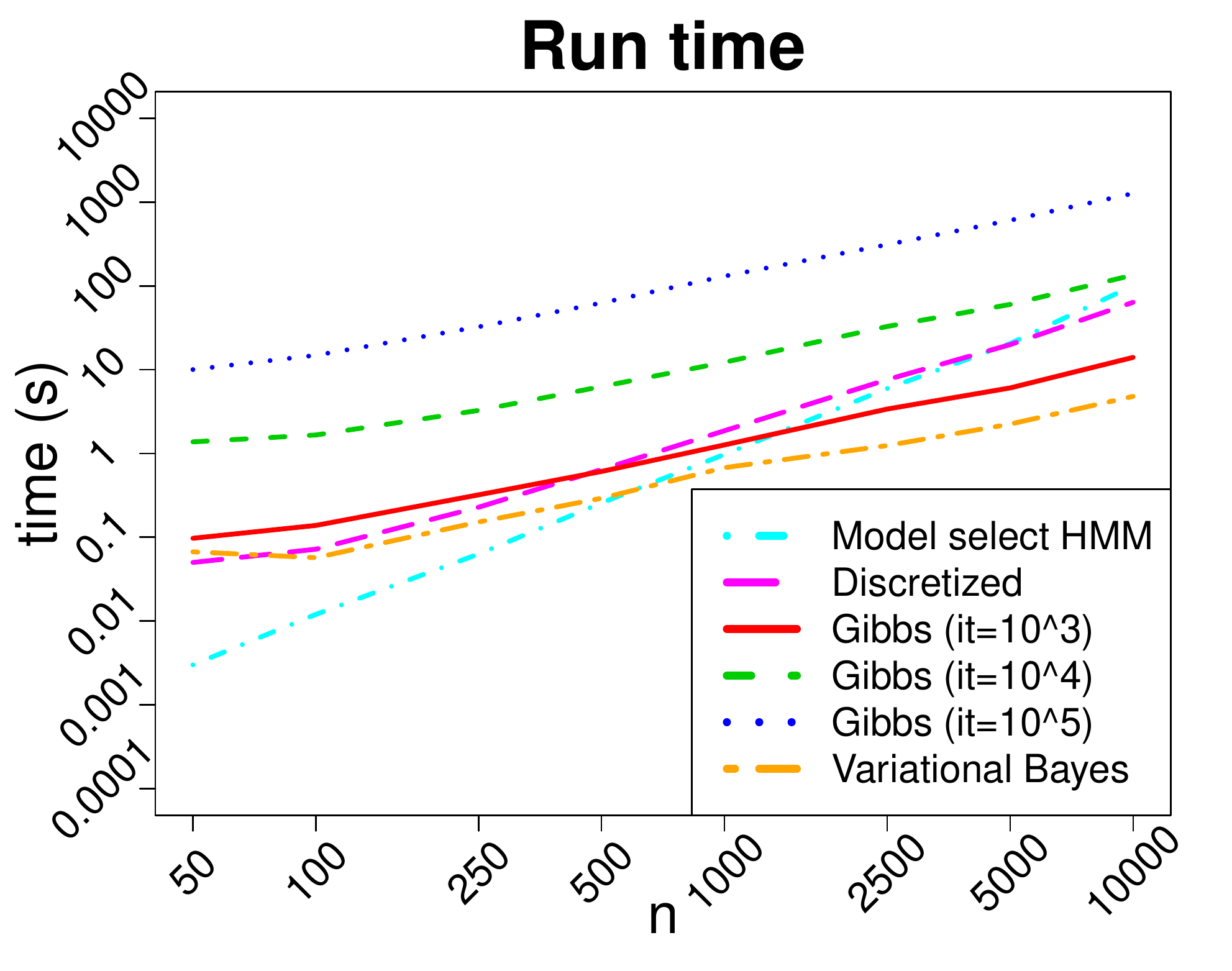}
  \end{minipage}
  \hfill
  \begin{minipage}[b]{0.48\textwidth}
    \includegraphics[width=\textwidth]{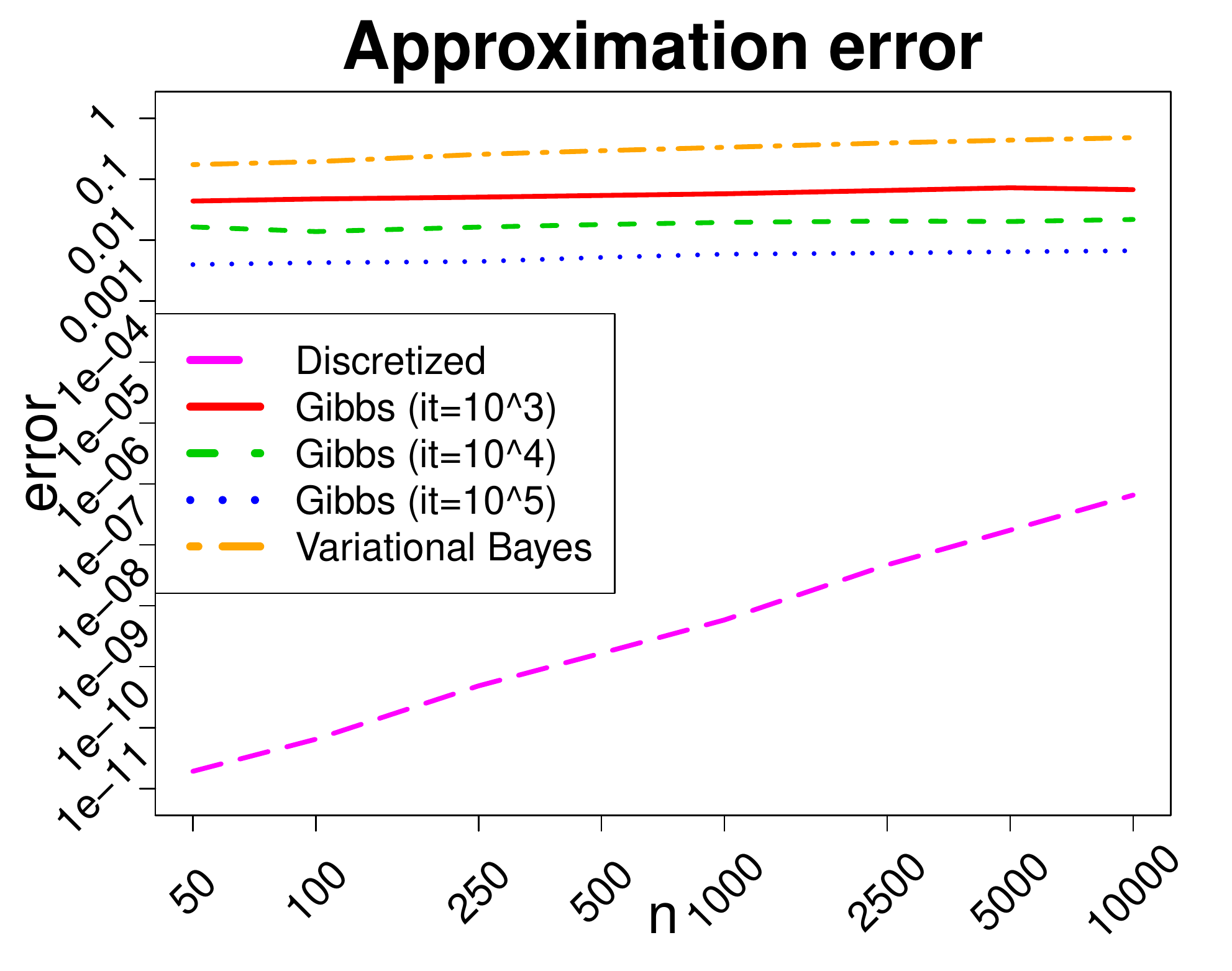}
  \end{minipage}
\caption{Run times and approximation errors for approximate algorithms.
An approximation error of $10^{-a}$ means that the algorithm is able to
calculate the correct answer up to $a$ decimal places.}
\label{fig:compare:approx:alg}
\end{figure}

\begin{table}[tbp]
\centering
\caption{Approximation errors compared to the exact HMM algorithm}
\label{table:approx}
\scriptsize{%
 \setlength{\tabcolsep}{1pt}%
 \renewcommand{\arraystretch}{1.3}%
\begin{tabular}{|l||l|l|l|l|l|l|l|}
 \hline
Method~\textbackslash~n&  100&250&500&1\,000&2\,500&5\,000&10\,000 \\
 \hline
 \hline
Discretized &
  $\mathbf{6.37\times 10^{-11}}$
 &$\mathbf{4.89\times 10^{-10}}$
 &$\mathbf{1.67\times 10^{-9}}$
 &$\mathbf{5.89\times 10^{-9}}$
 &$\mathbf{4.69\times 10^{-8}}$
 &$\mathbf{1.74\times 10^{-7}}$
 &$\mathbf{6.56\times 10^{-7}}$\\
 \hline
Gibbs  ($\text{it}=10^3$) & $4.58\times 10^{-2}$ & $4.76\times 10^{-2}$ & $ 5.03\times 10^{-2}$
& $5.41\times 10^{-2}$& $6.15\times 10^{-2}$ & $6.28\times 10^{-2}$& $6.96\times 10^{-2}$\\
 \hline
Gibbs ($\text{it}=10^4$)&$1.23\times 10^{-2}$ & $1.46\times 10^{-2}$& $1.63\times 10^{-2}$& $1.75\times 10^{-2}$ & $2.03\times 10^{-2}$& $2.07\times 10^{-2}$&$2.25\times 10^{-2}$\\
 \hline
Gibbs ($\text{it}=10^5$)   & $4.55\times 10^{-3}$& $5.05\times 10^{-3}$& $5.46\times 10^{-3}$&$5.63\times 10^{-3}$ &$5.71\times 10^{-3}$&$6.86\times 10^{-3}$&$6.88\times 10^{-3}$\\
 \hline
Variational Bayes  & $1.90 \times 10^{-1}$& $2.48\times 10^{-1}$& $2.93\times 10^{-1}$&$3.33\times 10^{-1}$&$3.91\times 10^{-1}$&$4.40\times 10^{-1}$&$4.81\times 10^{-1}$\\
 \hline
\end{tabular}}
\end{table}

In this section we measure the approximation error of a selection of
approximation algorithms by comparing them to the exact model selection
HMM algorithm, which serves as a benchmark for the correct answer. We
again consider the spike-and-slab prior with
$\Lambda_n=\betadist(1,n+1)$, but for simplicity we use standard
Gaussian slabs $g(x) = \frac{1}{\sqrt{2\pi}} e^{-x^2}$, since the
approximation methods are typically designed for this choice of slab
distribution. Our first approximation method is the discretization
algorithm from Section~\ref{sec:discrete}, which uses a deterministic
approximation. The discretization algorithm was again initialized
according to Corollary~\ref{cor:betaprior} with $m=20$. We further
consider a standard Gibbs sampler (with number of iterations
$\text{it}=10^3,10^4,10^5$, half of which are used as burn-in) and a
variational Bayes approximation. We consider the same test data as in
the preceding section. The only difference is that we stop at
$n=10\,000$ to limit the run times for the exact HMM algorithm
and the Gibbs sampler with $\text{it=}10^5$. Both the Gibbs sampler and
variational Bayes algorithm were implemented in R. For the latter we
used the component-wise variational Bayes algorithm
\cite{logsdon2010variational,titsias:2011,carbonetto:2012,ray:szabo:2019}. We
measure approximation error by computing $\max_i |q_{n,i} -
\tilde{q}_{n,i}|$, where $q_{n,i}$ is the exact slab probability
computed by the model selection HMM and $\tilde{q}_{n,i}$ is the slab
probability computed by the approximation. We run each non-deterministic
approximation method 5 times and report the average approximation error
along with the average run time of the algorithms. The results are
plotted in Figure~\ref{fig:compare:approx:alg} and shown numerically in
Table~\ref{table:approx}. 

One can see that the discretized version of the algorithm is very
accurate, with at least seven decimal places of precision throughout. It
approximately loses two decimal places of precision for every ten-fold
increase of $n$, so we can still expect it to be accurate up to five
decimal places for $n=100\,000$. We point out that its approximation
error includes both the mathematical approximation from
Section~\ref{sec:discrete} and the numerical error already studied
separately in Figure~\ref{fig:compare:exact:alg}. Since the
approximation error in Figure~\ref{fig:compare:approx:alg} is of the
same order as the numerical error in Figure~\ref{fig:compare:exact:alg},
we conclude that the numerical error dominates the mathematical
approximation error, so the discretization algorithm may be considered
an exact method for all practical purposes.

At the same time the Gibbs sampler and the variational Bayes method both
provide approximations of the posterior that are far less accurate.
Variational Bayes is only accurate up to one decimal place, although in
further investigations we did find that it provides a better
approximation if we look only at the non-zero coefficients, with an
approximation error of order $O(10^{-4})$. For the Gibbs sampler there
is no theory that tells us how many iterations we have to take to
achieve a certain degree of accuracy. We see here that the precision
strongly depends on the number of iterations, ranging from one to three
decimal places, but remains approximately constant with increasing $n$.
However, the run time for $\text{it}=10^5$ iterations would become
prohibitive for sample sizes much larger than the $n=10\,000$ we
consider.

\section{Differential Gene Expression for Ulcerative Colitis and Crohn's
Disease}
\label{sec:realworld}

In this section we compare our methods to various other frequently used
Bayesian approaches in the context of differential gene expression data.

\paragraph{Data}

We consider a data set from Burczynski et al.\ \cite{Burczynski:2006}
containing the gene expression levels of $n=22\,283$ genes in peripheral
blood mononuclear cells, with the raw data provided by the National
Center for Biotechnology Information.\footnote{Via the Gene Expression
Omnibus (GEO) website under dataset record number GDS1615. See
\href{http://www.ncbi.nlm.nih.gov/sites/GDSbrowser?acc=GDS1615}{www.ncbi.nlm.nih.gov/sites/GDSbrowser?acc=GDS1615}.} This is an
observational study, with microarray gene expression data on 26 subjects
who suffered from ulcerative colitis and 59 subjects with Crohn's
disease. We calculate Z-scores to identify
differences in average gene expression levels between the two disease
groups following the standard approach described by Quackenbush
\cite{quackenbush:2002}, which consists of dividing the difference of
the average log-transformed, normalized gene expressions for the two
groups by the standard error. More specifically, let us denote by
$U_{i,j}$ and $C_{i,j}$ the measured intensities of the $i$-th 
gene and $j$-th person with ulcerative colitis and Crohn's disease,
respectively. As a first step we normalize the intensities for each
patient, i.e.\ we take  $U_{i,j}'=U_{i,j}/\sum_{i} U_{i,j}$ and
$C_{i,j}'=C_{i,j}/\sum_{i} C_{i,j}$ for each gene $i$ and
patient $j$. Then the Z-score for the $i$-th gene is
computed as
\begin{align*}
Z_i=\frac{\overline{\log U_i'}-\overline{\log
C_i'}}{\sqrt{\sigma^2_{U',i}/26+\sigma^2_{C',i}/59}}
\qquad \text{for $i=1,\ldots, 22\,283$,}
\end{align*}
where 
\begin{align*}
\overline{\log U_i'}=\frac{1}{26}\sum_{j=1}^{26} \log U_{i,j}',
\quad 
\sigma^2_{U',i}=\frac{1}{25}\sum_{j=1}^{26}\big(\log U_{i,j}'-
\overline{\log U_i'} \big)^2,
\end{align*}
and $\overline{\log C_i'}$ and $\sigma^2_{C',i}$ are defined
accordingly.  Since it is assumed that the number of genes with a
different expression level between the two groups is small compared to
the total number of genes $n$, the data fit into the sparse normal
sequence model with $n=22\,283$.

\paragraph{Methods}

\begin{table}
  \centering
  \caption{Run time and number of selected genes on gene expression
  data. The reported run times for Empirical Bayes EBSparse
  and the Horseshoe are the averages over their runs.}
  \label{tab:RealWorld}
  \scriptsize{%
  \begin{tabular}{|l|r|r|}
    \hline
    Method & Run Time & Nr.\ of Genes Selected\\
    \hline
    \hline
    Variational Bayes  (varbvs)              &  20.95 minutes &   166\\
    Spike-and-Slab LASSO            &  0.01 seconds &  557\\
    Horseshoe (10 runs)              &  1.86 minutes & 571--583\\
    Empirical Bayes EBSparse (10 runs)          & 7.28 seconds & 592--604\\
    Discretized: $\betadist(1,n+1)$-binomial prior    &   24.47 seconds &   674\\
    HMM: $\betadist(1,n+1)$-binomial prior             &  2.06 minutes &  674\\ 
    Empirical Bayes JS               &  0.03 seconds & 3168\\
    HMM: $\betadist(1,1)$-binomial prior             &  2.00 minutes &  3169\\
    \hline
  \end{tabular}}
\end{table}
We compare the run times and the selected genes for the eight procedures
listed in Table~\ref{tab:RealWorld}. We consider the model selection HMM
algorithm for the $\betadist(1,n+1)$-binomial prior with Laplace slab
(with hyper-parameter $a=0.5$), and the
discretization algorithm from Corollary~\ref{cor:betaprior} with $m=20$,
which is a faster way to compute exactly the same results. Genes~$i$
with marginal posterior probability $q_{n,i} \geq 1/2$ are selected. For
comparison, we also consider the model selection HMM for the
$\betadist(1,1)$-binomial prior, which corresponds to using a uniform
prior $\Lambda_n$ on the mixing parameter $\alpha$. In this section, we
used the implementations of our algorithms from our R package
\cite{VanErvenDeRooijSzabo2019}, which is approximately 5 times faster
than the implementation from Section~\ref{sec:simulations}, because it
does not incur the overhead of tracking numerical accuracy using
interval arithmetic.

We compare to the empirical Bayes method of Johnstone and Silverman
\cite{johnstone2004}, which uses a spike-and-slab prior, but estimates
the mixing parameter $\alpha$ using empirical Bayes. The method does not
explicitly include a prior on $\alpha$, but we may interpret it as using
a uniform prior~$\Lambda_n$. We again use a Laplace slab (with the
default parameter $a=0.5$)  and select genes by hard thresholding at
marginal posterior probability $1/2$, as implemented in the R package
\cite{ebayesthresh:2017}.

We also include EBSparse, which is a fractional empirical Bayes
procedure proposed by Martin and Walker \cite{martin:2014}. It can be
interpreted as using a spike-and-slab prior with $\Lambda_n =
\betadist(1, \gamma n)$, but with Gaussian slabs $G_i =
\normaldist(Y_i,\tau^2)$ whose means depend on the data. Furthermore, in
the formula for the posterior the likelihood is tempered by raising it
to the power $\kappa$. We use the authors' R implementation
\cite{martin:EBSparse:2020}, with the recommended hyper-parameter
settings $\kappa=0.99$, $\gamma=0.25$, $\tau^2=100$, and $M=1000$ Monte
Carlo samples. As the sampler is
randomized, we run the algorithm $10$ times.

We further consider the Spike-and-Slab LASSO of Ro{\v{c}}kov{\'a}
\cite{rovckova:2018}, which computes the maximum a posteriori parameters
using Laplace distributions both for the spikes and for the slabs. As in
\cite[Section~6]{rovckova:2018}, we take the slab scale parameter to be
$\lambda_1 = 0.1$, and estimate the spike scale parameter $\lambda_0$
via the two-step procedure described there, for the $\betadist(1,n+1)$
hyper-prior on the mixing parameter. An R implementation called SSLasso
was provided by Ro{\v{c}}kov{\'a} \cite{rovckova:SSLsoftware:2018}.

We also add the Horseshoe estimator \cite{carvalho:2010} with the
Cauchy hyper-prior on its hyper-parameter $\tau$, truncated to the
interval $[1/n,1]$, as recommended by Van der Pas et al.\
\cite{pas:etal:2017}. We use the R package \cite{pas:et:al:2016}, with
its default Markov Chain Monte Carlo sampler settings of $1000$
iterations burn-in and $5000$ iterations after burn-in. Genes are
selected if their credible sets exclude zero \cite{pas:etal:2017}. As
the sampler is randomized, we run the algorithm $10$ times.

Finally, we compare with the variational Bayes algorithm (varbvs R-package)
described in \cite{carbonetto:2012}. Notably, this method uses Gaussian
slabs. The hyper-parameters (e.g.\ the variance of the prior and the
noise) are automatically fitted to the data. We set the tolerance and
maximum number of iterations to be $10^{-4}$ and $1000$, respectively. 

\begin{figure}[htb]
  \begin{minipage}[b]{0.9\textwidth}
    \includegraphics[width=\textwidth]{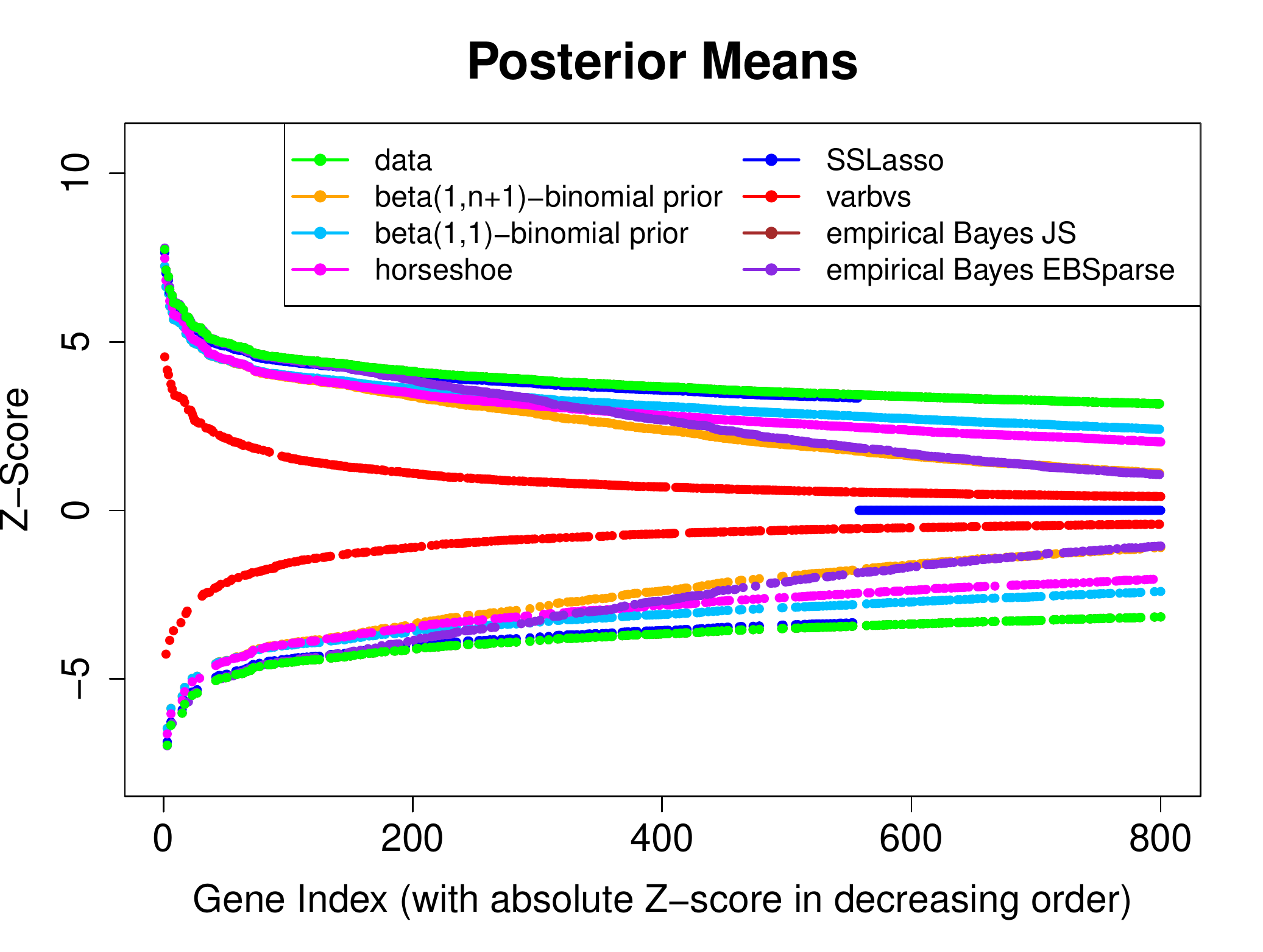}
  \end{minipage}
\caption{Posterior means/MAP estimates for the $800$ genes in the gene
expression data with largest Z-scores (in absolute value)}
\label{fig:real_world}
\end{figure}

\paragraph{Results}

Results are reported in Table~\ref{tab:RealWorld} and
Figure~\ref{fig:real_world}. Although we list run times to illustrate
computational feasibility, it is important to keep in mind that the
methods in this section compute different quantities, so their most
important difference lies in which genes they select. On this point, the
main conclusion is that the alternative methods give very different
results from using the exact Bayesian posterior for the model selection
prior.

All methods except the Horseshoe and EBSparse select genes in decreasing
order of the absolute values of $Z_i$. Genes are generally selected by
the Horseshoe and EBSparse in decreasing order of absolute value of
$Z_i$ as well, but with some swaps for genes for which the absolute
values are close to each other, so it appears that for all
sampling-based methods the sampler is suffering from limited precision,
as we also observed for the Gibbs samplers in
Section~\ref{sec:compare:approx}. The methods can be divided into three
main categories based on the number of genes they select: on one extreme
is the variational Bayes (varbvs) method, which provides the sparsest
solution; then the majority of methods select a number of genes between
$557$ and $674$; and finally at the other extreme are the Empirical
Bayes JS procedure and the $\betadist(1,1)$-binomial prior, which are
both based on the same prior and both select a very large number of
genes, making these two methods the most conservative. The lack of
sparsity induced by the $\betadist(1,1)$-binomial prior is perhaps not
surprising, given that it does not satisfy the exponential decrease
condition of \cite{castillo2012}. We study this further in
Section~\ref{sec:asymptotics}, where we compare different choices for
the hyper-parameters of the beta-binomial prior in simulations.

We further see that the Spike-and-Slab LASSO and the empirical Bayes JS 
procedures finish almost instantly. The EBSparse method takes several
seconds to run, as does the discretization algorithm. The Horseshoe and our exact model
selection HMM take approximately two minutes to run, while the variational Bayes
varbvs method requires a little over 20 minutes. Nevertheless, all methods are feasible even
for practitioners who would like to perform multiple similar
experiments, for example with different variations of the prior or slab
distributions. By contrast, we do not include the Castillo-Van der Vaart
algorithm with logarithmic representation, because based on
extrapolation of Figure~\ref{fig:compare:exact:alg} we expect it to take
around 20 days.

In Figure~\ref{fig:real_world} we also plot the posterior means (or, in
case of the Spike-and-Slab LASSO, the MAP estimator) and the $800$
largest Z-scores in absolute value. Since the posterior means for the
model selection HMM and the discretization algorithm are the same, we
label both as $\betadist(1,n+1)$-binomial in reference to the prior that
was used. We further note that the empirical Bayes JS estimates are
invisible behind the data points. We observe that the varbvs method
induces the heaviest shrinkage, followed first by the
$\betadist(1,n+1)$-binomial prior and the Empirical Bayes EBSparse
method, and then by the Horseshoe and the $\betadist(1,1)$-binomial
prior. The least shrinkage is applied by the empirical Bayes JS method,
which does not shrink the observed Z-scores very much (if at all). The
Spike-and-Slab LASSO is in a category of its own, because it is a MAP
estimator. It applies no shrinkage to the coefficients that are
selected, and sets all other coefficients to zero.

\section{Asymptotics of Spike-and-Slab Priors}
\label{sec:asymptotics}

The choice of the prior $\Lambda_n$ on the mixing hyper-parameter
$\alpha$ in spike-and-slab priors is considered to be highly relevant
for the behavior of the posterior. Castillo and Van der Vaart
\cite{castillo2012} recommend to use $\Lambda_n =
\betadist(\kappa,\lambda)$ with parameters $\kappa=1$ and $\lambda=n+1$.
This prior induces heavy penalization for dense models (models with
large sparsity parameter~$s$) and was shown to have optimal theoretical
properties. However, it is unknown whether such heavy penalization is
indeed necessary and whether even heavier penalization will result in
suboptimal behavior.

In this section we investigate the asymptotic behavior of the posterior
for different choices of the hyper-parameters $\kappa$ and $\lambda$
using our new exact algorithms, which can scale up to large sample
sizes. We consider: i)~the uniform prior with
$\kappa=1$ and $\lambda=1$, which is often considered a natural choice
\cite{scott:2010}; ii)~mild shrinkage, $\kappa=1$ and
$\lambda=\sqrt{n}$; iii)~the choice $\kappa=1$ and $\lambda=n+1$
recommended by Castillo and Van der Vaart; iv)~heavy shrinkage, $\kappa=1$ and $\lambda=n^2$ ;
and finally v) a sparsity-discouraging choice, $\kappa=n$ and
$\lambda=1$. We consider two experiments: $A_1$ and $A_2$. In both cases
the sample sizes range from $n=50$ to
$n=20\,000$. In Experiment $A_1$ we set the true sparsity level to
$s=10$ and consider uniformly distributed non-zero signal coefficients
between 1~and~10, i.e.\ $\theta_i\sim U(1,10)$ for $i\in \S$. In
Experiment $A_2$ the  true sparsity level is taken to be $s=\lceil
n^{1/3}\rceil$ and
the non-zero signal coefficients are set to $\theta_i=2\sqrt{2\log n}$
for $i\in \S$, which is a factor of $2$ above the detection threshold. In
Appendix~\ref{sec:simulation:extra} of the supplementary material we consider an
additional experiment $A_3$ that is similar to $A_1$ but with $s=25$ and
$\theta_i \sim U(5,10)$ for $i \in \S$, which gives similar results as
Experiment $A_1$.

\begin{table}[tbp]
\centering
\caption{$\ell_2$ distance of the posterior mean from the true signal in Experiment $A_1$}
\label{table:l2_1}
\scriptsize{%
 \setlength{\tabcolsep}{1pt}%
 \renewcommand{\arraystretch}{1.3}%
\begin{tabular}{|l||l|l|l|l|l|l|l|l|l|}
 \hline
Method~\textbackslash~n&  50&100&200&500&1\,000&2\,000&5\,000&10\,000&20\,000 \\
 \hline
 \hline
i) $\kappa=1$, $\lambda=1$ &4.52\,(0.64)&4.65\,(0.68)&4.84\,(0.75)&5.30\,(1.01)&5.18\,(0.74)&5.63\,(0.90)&5.95\,(0.77)&6.65\,(1.26)&6.38\,(0.98)\\
 \hline
ii) $\kappa=1$, $\lambda=\sqrt{n}$ &4.28\,(0.63)&4.46\,(0.69)&4.69\,(0.75)&5.20\,(1.01)&\textbf{5.13}\,(0.73)&5.57\,(0.88)&\textbf{5.92}\,(0.78)&6.62\,(1.27)&\textbf{6.36}\,(0.98)\\
 \hline
iii)  $\kappa=1$, $\lambda=n+1$&\textbf{4.09}\,(0.68) &\textbf{4.20}\,(0.82)&\textbf{4.43}\,(0.88)&\textbf{4.95}\,(1.11)&5.14\,(0.89)&\textbf{5.45}\,(0.99)&6.18\,(0.94)&\textbf{6.42}\,(1.67)&6.43\,(1.06)\\
 \hline
iv)  $\kappa=1$, $\lambda=n^2$  &5.04\,(1.05)&5.21\,(1.42)&6.55\,(1.48)&7.09\,(1.93)&7.87\,(1.21)&8.15\,(1.42)&9.45\,(1.55)&10.35\,(2.24)&10.43\,(2.16)\\
 \hline
v) $\kappa=n$, $\lambda=1$ & 5.71\,(0.59)&7.60\,(0.62)&10.63\,(0.54)&16.89\,(0.61)&23.39\,(0.61)&33.43\,(0.78)&52.93\,(0.52)&74.79\,(0.57)&105.84\,(0.55)\\
 \hline
\end{tabular}}
\end{table}

\begin{table}[tbp]
\centering
\caption{False discovery rate  in Experiment $A_1$}
\label{table:fdr1}
\scriptsize{%
 \setlength{\tabcolsep}{1pt}%
 \renewcommand{\arraystretch}{1.3}%
\begin{tabular}{|l||l|l|l|l|l|l|l|l|l|}
 \hline
Method~\textbackslash~n&  50&100&200&500&1\,000&2\,000&5\,000&10\,000&20\,000 \\
 \hline
 \hline
i) $\kappa=1$, $\lambda=1$ &0.36\,(0.14)&0.23\,(0.13)&0.19\,(0.11)&0.13\,(0.12)&0.05\,(0.09)&0.09\,(0.12)&0.08\,(0.09)&0.12\,(0.10)&0.08\,(0.09)\\
 \hline
ii) $\kappa=1$, $\lambda=\sqrt{n}$ &0.16\,(0.10)&0.16\,(0.11)&0.17\,(0.11)&0.11\,(0.11)&0.06\,(0.09)&0.09\,(0.12)&0.08\,(0.08)&0.11\,(0.10)&0.08\,(0.09)\\
 \hline
iii)  $\kappa=1$, $\lambda=n+1$&0.05\,(0.06) &0.03\,(0.05)&0.02\,(0.04)&0.02\,(0.04)&0.01\,(0.03)&0.01\,(0.02)&0.00\,(0.00)&0.01\,(0.03)&0.00\,(0.00)\\
 \hline
iv)  $\kappa=1$, $\lambda=n^2$  &{0.00}\,(0.00)&{0.00}\,(0.00)&{0.00}\,(0.00)&{0.00}\,(0.00)&{0.00}\,(0.00)&{0.00}\,(0.00)&{0.00}\,(0.00)&{0.00}\,(0.00)&{0.00}\,(0.00)\\
 \hline
v) $\kappa=n$, $\lambda=1$ & 0.80\,(0.00)&0.90\,(0.00)&0.95\,(0.00)&0.98\,(0.00)&0.99\,(0.00)&1.00\,(0.00)&1.00\,(0.00)&1.00\,(0.00)&1.00\,(0.00)\\
 \hline
\end{tabular}}
\end{table}

\begin{table}[tbp]
\centering
\caption{True positive rate  in Experiment $A_1$}
\label{table:tpr1}
\scriptsize{%
 \setlength{\tabcolsep}{1pt}%
 \renewcommand{\arraystretch}{1.3}%
\begin{tabular}{|l||l|l|l|l|l|l|l|l|l|}
 \hline
Method~\textbackslash~n&  50&100&200&500&1\,000&2\,000&5\,000&10\,000&20\,000 \\
 \hline
 \hline
i) $\kappa=1$, $\lambda=1$&0.91\,(0.11) &0.83\,(0.15)&0.86\,(0.08)&0.80\,(0.17)&0.76\,(0.13)&0.73\,(0.17)&0.73\,(0.13)&0.71\,(0.14)&0.66\,(0.16)\\
 \hline
ii) $\kappa=1$, $\lambda=\sqrt{n}$ &0.89\,(0.11)&0.83\,(0.14) &0.85\,(0.09)&0.78\,(0.17)&0.75\,(0.13)&0.73\,(0.17)&0.72\,(0.13)&0.71\,(0.15)&0.65\,(0.16)\\
 \hline
iii)  $\kappa=1$, $\lambda=n+1$&0.84\,(0.12) &0.76\,(0.11)&0.80\,(0.10)&0.74\,(0.18)&0.70\,(0.12)&0.68\,(0.18)&0.64\,(0.12)&0.64\,(0.19)&0.59\,(0.15)\\
 \hline
iv)  $\kappa=1$, $\lambda=n^2$  &0.71\,(0.14)&0.62\,(0.14)&0.59\,(0.16)&0.57\,(0.17)&0.52\,(0.13)&0.48\,(0.20)&0.44\,(0.15)&0.40\,(0.17)&0.38\,(0.16)\\
 \hline
v) $\kappa=n$, $\lambda=1$ &{1.00}\,(0.00) &{1.00}\,(0.00)&{1.00}\,(0.00)&{1.00}\,(0.00)&{1.00}\,(0.00)&{1.00}\,(0.00)&{1.00}\,(0.00)&{1.00}\,(0.00)&{1.00}\,(0.00)\\
 \hline
\end{tabular}}
\end{table}

We repeat each experiment 20 times and report the average $\ell_2$-error
between the posterior mean for $\theta$ and the true signal $\theta$ in
Table~\ref{table:l2_1} and Table~\ref{table:l2_3} for Experiments $A_1$
and $A_2$, respectively. In Tables~\ref{table:fdr1}, \ref{table:tpr1},
\ref{table:fdr3} and \ref{table:tpr3} we also report the average false
discovery rates and the average true positive rates. Standard deviations
are provided in parentheses in all cases.

\begin{table}[tbp]
\centering
\caption{$\ell_2$ distance of the posterior mean from the true signal in Experiment $A_2$}
\label{table:l2_3}
\scriptsize{%
 \setlength{\tabcolsep}{1pt}%
 \renewcommand{\arraystretch}{1.3}%
\begin{tabular}{|l||l|l|l|l|l|l|l|l|l|}
 \hline
Method~\textbackslash~n&  50&100&200&500&1\,000&2\,000&5\,000&10\,000&20\,000 \\
 \hline
 \hline
i) $\kappa=1$, $\lambda=1$ &3.03\,(0.98)&3.77\,(0.70)&3.9\,(0.58)&4.48\,(0.96)&4.68\,(0.51)&5.65\,(0.84)&7.13\,(0.81)&7.70\,(1.21)&8.68\,(0.82)\\
 \hline
ii) $\kappa=1$, $\lambda=\sqrt{n}$ &2.79\,(0.91)&3.56\,(0.68)&3.80\,(0.59)&4.32\,(0.93)&4.55\,(0.51)&5.53\,(0.83)&7.02\,(0.81)&7.61\,(2.15)&8.59\,(1.21)\\
 \hline
iii)  $\kappa=1$, $\lambda=n+1$&\textbf{2.41}\,(0.88) &\textbf{3.01}\,(0.68)&\textbf{3.07}\,(0.74)&\textbf{3.48}\,(0.79)&\textbf{3.64}\,(0.62)&\textbf{4.27}\,(0.63)&\textbf{5.33}\,(0.78)&\textbf{5.68}\,(0.94)&\textbf{6.32}\,(0.77)\\
 \hline
iv)  $\kappa=1$, $\lambda=n^2$  &3.43\,(1.82)&3.63\,(2.06)&3.89\,(1.86)&5.13\,(2.51)&4.49\,(1.90)&5.08\,(2.27)&5.62\,(1.42)&6.27\,(2.15)&6.42\,(1.70)\\
 \hline
v) $\kappa=n$, $\lambda=1$ & 5.31\,(0.84)&7.69\,(0.48)&10.67\,(0.52)&16.81\,(0.71)&23.56\,(0.56)&33.59\,(0.50)&53.00\,(0.44&74.92\,(0.53)&105.8\,(0.54)\\
 \hline
\end{tabular}}
\end{table}

\begin{table}[tbp]
\centering
\caption{False discovery rate  in Experiment $A_2$}
\label{table:fdr3}
\scriptsize{%
 \setlength{\tabcolsep}{1pt}%
 \renewcommand{\arraystretch}{1.3}%
\begin{tabular}{|l||l|l|l|l|l|l|l|l|l|}
 \hline
Method~\textbackslash~n&  50&100&200&500&1\,000&2\,000&5\,000&10\,000&20\,000 \\
 \hline
 \hline
i) $\kappa=1$, $\lambda=1$ &0.25\,(0.22)&0.22\,(0.17)&0.15\,(0.12)&0.12\,(0.09)&0.10\,(0.08)&0.11\,(0.08)&0.11\,(0.06)&0.12\,(0.06)&0.10\,(0.03)\\
 \hline
ii) $\kappa=1$, $\lambda=\sqrt{n}$ &0.15\,(0.16)&0.18\,(0.14)&0.14\,(0.11)&0.11\,(0.09)&0.08\,(0.08)&0.10\,(0.08)&0.10\,(0.06)&0.11\,(0.06)&0.10\,(0.03)\\
 \hline
iii)  $\kappa=1$, $\lambda=n+1$&0.01\,(0.04) &0.06\,(0.08)&0.03\,(0.07)&0.04\,(0.06)&0.01\,(0.03)&0.02\,(0.05)&0.02\,(0.03)&0.02\,(0.02)&0.01\,(0.02)\\
 \hline
iv)  $\kappa=1$, $\lambda=n^2$  &0.00\,(0.00)&0.00\,(0.00)&0.00\,(0.00)&0.00\,(0.00)&0.00\,(0.00)&0.00\,(0.00)&0.00\,(0.00)&0.00\,(0.00)&0.00\,(0.00)\\
 \hline
v) $\kappa=n$, $\lambda=1$ & 0.92\,(0.00)&0.95\,(0.00)&0.97\,(0.00)&0.98\,(0.00)&0.99\,(0.00)&0.99\,(0.00)&1.00\,(0.00)&1.00\,(0.00)&1.00\,(0.00)\\
 \hline
\end{tabular}}
\end{table}

\begin{table}[tbp]
\centering
\caption{True positive rate in Experiment $A_2$}
\label{table:tpr3}
\scriptsize{%
 \setlength{\tabcolsep}{1pt}%
 \renewcommand{\arraystretch}{1.3}%
\begin{tabular}{|l||l|l|l|l|l|l|l|l|l|}
 \hline
Method~\textbackslash~n&  50&100&200&500&1\,000&2\,000&5\,000&10\,000&20\,000 \\
 \hline
 \hline
i) $\kappa=1$, $\lambda=1$&1.00\,(0.00) &1.00\,(0.00)&1.00\,(0.00)&1.00\,(0.00)&1.00\,(0.00)&1.00\,(0.00)&1.00\,(0.00)&1.00\,(0.00)&1.00\,(0.00)\\
 \hline
ii) $\kappa=1$, $\lambda=\sqrt{n}$ &1.00\,(0.00) &1.00\,(0.00)&1.00\,(0.00)&1.00\,(0.00)&1.00\,(0.00)&1.00\,(0.00)&1.00\,(0.00)&1.00\,(0.00)&1.00\,(0.00)\\
 \hline
iii)  $\kappa=1$, $\lambda=n+1$&1.00\,(0.00) &1.00\,(0.00)&1.00\,(0.00)&1.00\,(0.00)&1.00\,(0.00)&1.00\,(0.00)&1.00\,(0.00)&1.00\,(0.00)&1.00\,(0.00)\\
 \hline
iv)  $\kappa=1$, $\lambda=n^2$  &0.89\,(0.19)&0.94\,(0.09)&0.94\,(0.08)&0.95\,(0.09)&0.98\,(0.05)&0.99\,(0.03)&0.99\,(0.02)&0.99\,(0.02)&1.00\,(0.01)\\
 \hline
v) $\kappa=n$, $\lambda=1$ &1.00\,(0.00) &1.00\,(0.00)&1.00\,(0.00)&1.00\,(0.00)&1.00\,(0.00)&1.00\,(0.0v0)&1.00\,(0.00)&1.00\,(0.00)&1.00\,(0.00)\\
 \hline
\end{tabular}}
\end{table}

In Experiment $A_1$ we see that the $\ell_2$-error is not very sensitive
to the choice of hyperparameters: the uniform prior i), the mild
shrinkage ii), and Castillo and Van der Vaart's recommendation iii) all
perform comparably. Only the heavy shrinkage iv) is introducing too high
penalization, especially for large models. Unsurprisingly, the choice of
hyper-parameters v) is also substantially worse than the others, because
it expresses exactly the wrong type of prior assumptions by heavily
penalizing sparse models. In Experiment $A_2$ we see that the best
hyper-parameters are Castillo and Van der Vaart’s recommendation iii)
and the heavy shrinkage iv), with the latter having a large
variability in performance. Hyper-parameter choices i) and ii) are
introducing no or only mild penalization for large models and indeed are
also observed to have somewhat worse performance than choices iii) and
iv), with the difference getting more pronounced for larger sample
sizes. Finally, as in Experiment $A_1$, the hyper-parameter setting
v)~is the worst by far.

We also study the false discovery rate (FDR) and true positive rate
(TPR) of the spike-and-slab priors (relatedly, see
\cite{castillo:roquain:2018} for the theoretical underpinning of FDR
control with empirical Bayes spike-and-slab priors). Unsurprisingly, the FDR is smallest in both experiments in case of heavy shrinkage
v), but almost equally good rates are obtained for the recommended
choice iii). Mild ii) or no  i) shrinkage result in somewhat worse FDR,
while the sparsity discouraging setting v) essentially selects all the
noise. In Experiment~$A_1$ the best TPR is obtained, not surprisingly,
by setting v), which conservatively selects everything. Hyper-parameter
choices i) and ii) perform comparably well, closely followed by iii),
while the heavy shrinkage method iv) is substantially worse. In
Experiment $A_2$ all hyper-parameter settings perform equally well,
except for the heavy shrinkage iv), which is slightly worse. The
good performance of the methods is due to the relatively high value
($2\sqrt{2\log n}$) for the non-zero signal coefficients, which
lies above the detection threshold $\sqrt{2\log n}$.

We conclude that, overall, the recommended choice iii) indeed appears to
have an advantage over the alternatives, and that even heavier
penalization as in choice iv) is harmful.

The above simulation study is just one example of how our exact
algorithms can be used to study asymptotic properties of model selection
priors, and more specifically spike-and-slab priors. Another possible
application not considered here would, for instance, be to study the
accuracy of Bayesian uncertainty quantification (see
\cite{castillo:szabo:2018} for frequentist coverage of Bayesian credible
sets resulting from empirical Bayes spike-and-slab priors). 

\section{Discussion}\label{sec:discuss}

We have proposed fast and exact algorithms for computing the Bayesian
posterior distribution corresponding to model selection priors
(including spike-and-slab priors as a special case) in the sparse normal
sequence model. Since the normal sequence model corresponds to linear
regression with identity design, the question arises whether the derived
algorithms can be extended to sparse linear regression with more general
designs or other more complex models. We first note that all methods are
agnostic about where the conditional densities of the spikes $p(Y_i\mid
B_i=0) = \phi(Y_i)$ and the slabs $p(Y_i \mid B_i=1) = \psi(Y_i)$ come
from. It is therefore trivial to extend them to any model that replaces
the distribution of $Y_i$ given $\S$ by
\[
  Y_i \mid \S \sim 
    \begin{cases}
      \psi_i & \text{if $i \in \S$,}\\
      \phi_i & \text{otherwise,}
    \end{cases}
\]
for any densities $\psi_i$ and $\phi_i$. (In fact, this is already
supported by our R package \cite{VanErvenDeRooijSzabo2019}.)
Such extensions make it possible to
easily handle other noise models for $\varepsilon_i$ or general diagonal
designs; and, as pointed out by a referee, it also allows incorporating
a non-atomic prior on $\theta_i$ in case $i \not \in \S$. We further
anticipate that extensions to general sparse design matrices may be
possible by generalizing the HMM from Section~\ref{sec:HMM} to more
general Bayesian networks and applying a corresponding inference
algorithm to compute marginal posterior probabilities. However, for
non-sparse design matrices the extension would be very challenging, if
possible at all, because the Bayesian network of the hidden states could
become fully connected. An interesting intermediate case is studied by
Papaspiliopoulos and Rossell \cite{PapaspiliopoulosRossell2017}, who
consider best-subset selection for block-diagonal designs. For the
normal sequence model, their assumptions amount to the requirement that
$\Lambda_n$ is a point-mass on a single $\alpha$, and they point out
that in this case ``best-subset selection becomes trivial.'' For
non-diagonal designs their results are non-trivial, because they are
able to integrate over a continuous hyper-prior on the variance
$\sigma^2$ of the noise $\varepsilon_i$. In contrast, we assume fixed
$\sigma^2$, which we then take to be $\sigma^2=1$ without loss of
generality. Our methods can be used to calculate the marginal likelihood
$p(Y\mid \sigma^2)$ without further computational overhead, so it would
be possible to run them multiple times to incorporate a discrete prior
on a grid of values for $\sigma^2$, but it is not obvious if our results
can be extended to continuous priors over $\sigma^2$. Exploration of
these directions is left for future work. 

Even without extending our methods to full linear regression or
continuous priors on $\sigma^2$, we believe that they are already very
useful as a benchmark procedure: any approximation technique for general
linear regression may be applied to the special case of sparse normal
sequences and its approximation error computed as in
Section~\ref{sec:compare:approx}. If a method does not work well in this
special case, then certainly we cannot trust it for more general
regression. The existence of such a benchmark method is very important,
since, for instance, there are no available diagnostics to determine
whether Markov Chain Monte Carlo samplers have converged to their
stationary distribution or if they have explored a sufficient proportion
of the models in the model space.

We have also explored the exact connection between general model
selection priors and the more specific spike-and-slab priors. Since for
spike-and-slab priors one can construct faster algorithms, it is useful
to know which model selection priors can be represented in this form.
The proof of our result amounts to a finite sample version of de
Finetti's theorem for a particular subclass of exchangeable
distributions, which may be of interest in its own right.

\section*{Acknowledgements}

We would like to thank Steven de Rooij for performing the experiments
with Karatsuba's algorithm reported in Appendix~\ref{sec:speedupCVdV} in
the supplemental material, for suggesting the use of interval arithmetic
in the experiments, and for detailed discussions underlying
Theorem~\ref{thm:discrete}. We further thank Isma{\"e}l Castillo for
providing the R code for the experiments in \cite{castillo2012} and
Veronika Ro{\v{c}}kov{\'a} for providing an R implementation of the
Spike-and-Slab LASSO for the normal sequence model
\cite{rovckova:SSLsoftware:2018}.

\section*{Supplementary Material}

The supplement contains a review of the exact algorithm by Castillo and
Van der Vaart and a discussion on how to perform all computations in a
logarithmic representation. It further includes an additional variation
on Experiment $A_1$ from Section~\ref{sec:asymptotics}. Finally, the
proofs for all theorems and the examples from
Section~\ref{sec:connection} are also given in the supplement.

\bibliographystyle{abbrvnat}

\bibliography{rates}

\clearpage

\appendix

\begin{center}
  \LARGE
  \bfseries
  Supplementary Material
\end{center}

\section{The Castillo-Van der Vaart Algorithm}
\label{sec:CvdV}
In this section we first recall the Castillo-Van der Vaart algorithm \cite{castillo2012}. 
Straight-forward implementation of the algorithm fails for
sample sizes larger than $n\geq 500$, because the intermediate results
exceed the maximum range that can be numerically represented.
Fortunately, this can be resolved by performing all computations in a
logarithmic representation, which we discuss second. The bottleneck then
becomes the algorithm's computational complexity, because it requires
$O(n^3)$ steps, which is prohibitive for large $n$. At the end of this
appendix we discuss two possible speed ups of the algorithm based on
fast polynomial multiplication and long division, respectively, and show
that neither of them works well practice. 

\subsection{Description of the algorithm}

The key ingredient of the Castillo-Van der Vaart algorithm is their
observation that, for any $s \in \{0,1,\ldots,n\}$ and any sequences of
numbers $a = (a_1,\ldots,a_n)$ and $b = (b_1,\ldots,b_n)$, the sum
\[
  C_s(a,b) = \sum_{|\S|=s} \prod_{i \in \S} a_i \prod_{i \neq \S} b_i
\]
is the coefficient of $Z^s$ in the polynomial
\begin{equation}\label{eqn:polynomial}
  Z \mapsto \prod_{i=1}^n (a_i Z + b_i).
\end{equation}
All coefficients of this polynomial can be computed in $O(n^2)$
operations by computing the products term by term, which is much faster
than explicitly summing over the exponentially many subsets of size $s$.
This observation allows Castillo and Van der Vaart to compute the
Bayesian marginal likelihood as follows:
\begin{equation}\label{eqn:marginallikelihood}
  Q_n = \sum_{s=0}^n \frac{\pi_n(s)}{\binom{n}{s}} \sum_{|\S|=s}
  \prod_{i \in \S} \psi(Y_i) \prod_{i \neq \S} \phi(Y_i)
      = \sum_{s=0}^n \frac{\pi_n(s)}{\binom{n}{s}} C_s(\vpsi,\vphi),
\end{equation}
with $\vpsi = (\psi(Y_1),\ldots,\psi(Y_n))$ and $\vphi =
(\phi(Y_1),\ldots,\phi(Y_n))$. \ The binomial coefficients can be precomputed in
$O(n)$ time using the recursion $\binom{n}{s} =
\binom{n}{s-1}(n-s+1)/s$.\footnote{For extra numerical precision it is
sometimes recommended to compute the binomial coefficients using
Pascal's triangle, but this takes $O(n^2)$ steps and the precision of
these coefficients is not the limiting factor of the algorithm.}
Assuming that $\pi_n(s)$ can be evaluated efficiently, computing the sum in
\eqref{eqn:marginallikelihood} then takes another $O(n)$ steps, which
means that the computation of the coefficients $C_s(\vpsi,\vphi)$ is the
dominant factor and all together $Q_n$ can be computed in $O(n^2)$
steps.

The same idea can be used again to compute the marginal posterior
probabilities
\begin{align}
  q_{n,i}
    &= \frac{1}{Q_n} \sum_{s=1}^n \frac{\pi_n(s)}{\binom{n}{s}} \psi(Y_i)
    \sum_{\substack{|\S|=s\\i \in \S}} \prod_{\substack{j \in \S\\j \neq i}} \psi(Y_j) \prod_{j
    \neq \S} \phi(Y_j) \notag\\
    &= \frac{1}{Q_n} \sum_{s=1}^n \frac{\pi_n(s)}{\binom{n}{s}} 
    \sum_{|\S|=s} \prod_{j \in \S} \vpsi_j \prod_{j \neq \S} \vphi^i_j
    = \frac{1}{Q_n} \sum_{s=1}^n \frac{\pi_n(s)}{\binom{n}{s}}
    C_s(\vpsi,\vphi^i),
  \label{eqn:marginalpostCvdV}
\end{align}
where $\vpsi$ is as before, and $\vphi^i$ equals $\vphi$ except that the
$i$-th component is replaced by~$0$. When $Q_n$ has been precomputed,
calculating $q_{n,i}$ takes $O(n^2)$ operations, just like computing
\eqref{eqn:marginallikelihood}. Repeating for all $n$ marginal posterior
probabilities $q_{n,1},\ldots,q_{n,n}$ therefore takes $O(n^3)$
operations in total.

\subsection{Logarithmic Representation}
\label{sec:logarithmic}

The Castillo-Van der Vaart algorithm (in its basic form described above)
works well for small sample sizes, but, as demonstrated in
Section~\ref{sec:simulations}, starts to fail for~$n$ larger than
roughly $500$. The reason is not computation time, which is still very
reasonable for these sample sizes, but the fact that the coefficients
$C_s(\vpsi,\vphi)$ and $C_s(\vpsi,\vphi^i)$ can take values ranging from
exponentially small in $n$ to exponentially large, and will therefore
underflow to zero or overflow to infinity when represented in the
standard double-precision floating-point format.

This range issue, however, can be resolved by using the following trick:
instead of the original quantities, we only compute the logarithms of
the (nonnegative) numbers $C_s(\vpsi,\vphi)$, $C_s(\vpsi,\vphi^i)$,
$\binom{n}{s}$ and $\pi_n(s)$, and we calculate
\eqref{eqn:marginallikelihood} and \eqref{eqn:marginalpostCvdV} using
these logarithmic representations.

Of course we cannot then, as an intermediate step, ever exponentiate our
numbers, so some care is needed when performing basic arithmetic. Given arbitrary
numbers $x = \ln a$ and $y = \ln b$, multiplication and division without
exponentiating are straightforward:
\begin{align*}
  \ln (a b) &= x + y && (\text{multiplication})\\
  \ln (a/b) &= x - y && (\text{division}).
\end{align*}
For addition and subtraction, we avoid direct exponentiation as follows:
assume without loss of generality that $x \geq y$; then
\begin{align*}
  \ln (a + b) &= x + \ln(1 + e^{y - x}) && (\text{addition}),\\
  \ln (a - b) &= x + \ln(1 - e^{y - x}) && (\text{subtraction}).
\end{align*}
Since $y - x \leq 0$ by assumption, these calculations can never
overflow. It is still possible that $\exp(y - x)$ underflows to $0$ if
$x \gg y$, but in that case the result will be $x$, which is very
accurate. (See e.g.\ \citep{Cook2011} for a similar discussion.) We
apply the rules above for $a,b \in [0,\infty]$ with the conventions
$\ln(\infty) = \infty$ and $\ln(0) = -\infty$ whenever the respective
operations are well-defined. For addition, there are therefore two cases
that require special care: if $x = y \in \{-\infty,\infty\}$, then $y-x$
is not defined, but $\ln(a + b)$ still makes sense; and for subtraction
$\ln(a - b)$ also makes sense for the case $x=y=-\infty$. These should
therefore be handled separately by defining
\begin{align*}
  \ln(a + b) &= 
    \begin{cases}
      \infty & \text{if $x = y = \infty$},\\
      -\infty & \text{if $x = y = -\infty$},
    \end{cases}\\
  \ln(a - b) &=
      -\infty \qquad \text{if $x = y = -\infty$.}
\end{align*}
The logarithmic representations and arithmetical rules described above
resolve the numerical accuracy issue by greatly extending the range of
representable values. One may wonder, however, whether, in the process,
we have not reduced the precision with which numbers are being stored by
too much. Luckily, this turns out not to be the case. In
Section~\ref{sec:simulations} we perform extensive experiments, which
confirm that, indeed, the resulting algorithm achieves high numerical
accuracy.

\subsection{Speeding up the Castillo-Van der Vaart Algorithm}
\label{sec:speedupCVdV}

In this subsection we investigate ways of speeding up the
Castillo-Van der Vaart algorithm. We consider two promising approaches
based on fast polynomial multiplication and long division, which,
surprisingly, both turn out to have severe limitations.

\paragraph{Fast Polynomial Multiplication}

Castillo and Van der Vaart \cite{castillo2012} point out
that polynomial multiplication, which naively takes $O(n^2)$ steps, is
actually possible in $O(n \ln^k n)$ steps for suitable $k$ (they
suggest $k=2$), which would allow computing all marginal posterior
probabilities $q_{n,1}, \ldots, q_{n,n}$ in $O(n^2 \ln^k n)$ steps.
Indeed, one possible approach is to recursively split
\eqref{eqn:polynomial} into $O(\ln n)$ multiplications of two
polynomials of equal size, and use an advanced algorithm for general
polynomial multiplication like the Toom-Cook algorithm \citep{knuth:1997},
which requires $O(n 2^{\sqrt{2\ln n}}\ln n)$ steps, or the
Sch{\"o}nhage-Strassen algorithm \citep{schonhage:1971}, which
requires $O(n \ln n \ln \ln n)$ steps. However, the constants in
these asymptotic rates are prohibitive and therefore the benefits of
these advanced algorithms only kick in for very large $n$. We have
experimented with the Karatsuba algorithm \citep{karatsuba:1962}, which
is a simpler special case of Toom-Cook, and at best obtained a factor of
$10$ speed-up for $n \leq 10^6$ when computing polynomials like
\eqref{eqn:polynomial}, which is minor compared to a factor of $n$
speed-up when $n=10^6$. We therefore do not consider the gains
sufficient to warrant the extra algorithmic complexity of using these
more advanced algorithms. Furthermore, there is no potential use for the
case $n > 10^6$ either, because then $O(n^2 \ln^k n)$ steps for the
total algorithm is already prohibitive regardless of the exact constants
in the polynomial multiplication subroutine.

\paragraph{Long Division}

We next describe a second attempt at speeding up the Castillo-Van der
Vaart algorithm, initially suggested by \citet{Castillo2017}, which is
based on long division. The main observation is that, for any $i$, the
polynomial \eqref{eqn:polynomial} for $(\vpsi,\vphi^i)$ differs from the
polynomial for $(\vpsi,\vphi)$ only in the $i$-th factor. Since we will
compute the coefficients $C_s(\vpsi,\vphi)$ of the polynomial for
$(\vpsi,\vphi)$ anyway (in the process of calculating $Q_n$), we can
divide off the $i$-th factor using long division for polynomials to
obtain the vector of coefficients $x = (x_0,\ldots,x_{n-1})$ such that
\begin{equation}\label{eqn:longdivision}
  \prod_{j=0}^{n-1} x_j Z^j
    = \prod_{\substack{j = 0,\ldots,n\\j \neq i}} (\vpsi_j Z + \vphi_j)
    = \frac{\prod_{j=0}^{n} C_j(\vpsi,\vphi) Z^j}{\vpsi_i Z + \vphi_i}
  \qquad \text{for all $Z$.}
\end{equation}
As explained below, this takes $O(n)$ steps. Multiplying the polynomial
$\prod_{j=0}^{n-1} x_j Z^j$ by $(\vpsi_i Z + \vphi^i_i)$ then takes
another $O(n)$ steps, and consequently we can compute the coefficients
$C_s(\vpsi,\vphi^i)$ needed in \eqref{eqn:marginalpostCvdV}, in $O(n)$
steps instead of the $O(n^2)$ steps we required before. Doing this for
$i=1,\ldots,n$ therefore takes $O(n^2)$ steps in total, which is a
speed-up of a factor $n$ compared to the original Castillo-Van der Vaart
algorithm.

As demonstrated in Section~\ref{sec:simulations}, the improvement from
$O(n^3)$ to $O(n^2)$ operations provides a major speed-up.
Unfortunately, however, we also show in Section~\ref{sec:simulations}
that performing long division (i.e.\ solving \eqref{eqn:longdivision}
for $x$) is numerically so unstable that the results can become
unreliable, even when using the logarithmic representation from
Appendix~\ref{sec:logarithmic}. It is therefore worth elaborating on how
we solve \eqref{eqn:longdivision}.

Solving this identity for $x$ amounts to solving the
overconstrained\footnote{The system is overconstrained because we know
the remainder of the long division will be zero.} linear system $B x =
a$ for 
\begin{align*}
    B&= 
    \begin{pmatrix}
      \vphi_i\\
      \vpsi_i & \vphi_i\\
          & \vpsi_i & \vphi_i\\
          &&\ddots&\ddots\\
          &&& \vpsi_i & \vphi_i\\
          &&&& \vpsi_i
    \end{pmatrix},
    & a&=
    \begin{pmatrix}
      C_0(\vpsi,\vphi)\\
      C_1(\vpsi,\vphi)\\
      C_2(\vpsi,\vphi)\\
      \vdots\\
      C_{n-1}(\vpsi,\vphi)\\
      C_n(\vpsi,\vphi)
    \end{pmatrix}.
\end{align*}
After dropping any row from this system of equalities, it can be solved
in $O(n)$ steps using back-substitution. We opt to drop the first row,
which makes the resulting procedure identical to long division. The
trouble with this approach is that it performs many divisions, which
translate into subtractions in the logarithmic representation, and
subtractions of two numbers of similar size can quickly lose numerical
precision. These errors accumulate while calculating the coefficients of
$x$ and hence the coefficients  that are calculated at the end of the
procedure are unreliable. We have therefore experimented with
alternatives like dropping the last or middle rows, or calculating
different parts of $x$ based on dropping different rows. We have also
tried an iterative refinement approach that apparently goes back to the
early days of computing in the 1940s \cite[p.\,184]{Higham2002}: here
$x^1$ is the solution initially computed and we repeatedly refine our
answer according to $x^{t+1}$ = $x^t + y^t$, where $y^t$ fits the
residuals: $B y^t$ = $a - B x^t$. This may still be computationally
attractive for small $t$, like e.g.\ $t \leq 5$. Although these
variations could sometimes postpone the problem to slightly larger $n$,
none of them has lead to a way to resolve it.

\section{Additional simulation study}\label{sec:simulation:extra}
In this section we provide an additional simulation study to Section \ref{sec:asymptotics}. In Experiment $A_3$ we consider the same hyper-parameter choices  of the prior $\Lambda_n=\betadist(\kappa,\lambda)$, i.e.
i) $\kappa=1$ and $\lambda=1$;  ii) 
$\kappa=1$ and $\lambda=\sqrt{n}$; iii) heavy shrinkage, $\kappa=1$ and $\lambda=n+1$; iv) $\kappa=1$ and $\lambda=n^2$ ;
and v) $\kappa=n$ and
$\lambda=1$. We take sample sizes ranging from $n=50$ to
$n=20\,000$,  choose the true sparsity level to be $s=25$ and consider uniformly distributed non-zero signal
coefficients between 5 and 10, i.e.\ $\theta_i\sim U(5,10)$ for $i\in \S$. 

We repeat each experiment 20 times and report the average $\ell_2$-error between the posterior mean for $\theta$ and the true signal $\theta$, the false discovery rates and the true positive rates, and their standard deviation in parenthesis.  The results are collected in
Tables~\ref{table:l2_2}, \ref{table:fdr2}, and \ref{table:tpr2}, respectively. We can conclude that we obtain comparable results to Section  \ref{sec:asymptotics}.

\begin{table}[tbp]
\centering
\caption{$\ell_2$ distance of the posterior mean from the true signal in Experiment $A_3$}
\label{table:l2_2}
\scriptsize{%
 \setlength{\tabcolsep}{1pt}%
 \renewcommand{\arraystretch}{1.3}%
\begin{tabular}{|l||l|l|l|l|l|l|l|l|l|}
 \hline
Method~\textbackslash~n&  50&100&200&500&1\,000&2\,000&5\,000&10\,000&20\,000 \\
 \hline
 \hline
i) $\kappa=1$, $\lambda=1$ &6.67\,(1.05)&7.21\,(0.87)&7.59\,(0.82)&7.65\,(1.19)&8.51\,(1.15)&8.44\,(0.82)&8.90\,(1.00)&9.03\,(1.14)&9.53\,(1.65)\\
 \hline
ii) $\kappa=1$, $\lambda=\sqrt{n}$ &6.37\,(1.07)&6.77\,(0.90)&7.27\,(0.79)&7.45\,(1.18)&8.34\,(1.16)&8.31\,(0.83)&8.80\,(1.01)&8.96\,(1.14)&9.47\,(1.65)\\
 \hline
iii)  $\kappa=1$, $\lambda=n+1$&\textbf{6.04}\,(1.08) &\textbf{6.23}\,(1.08)&\textbf{6.41}\,(0.93)&\textbf{6.61}\,(1.24)&\textbf{7.49}\,(1.46)&\textbf{7.48}\,(1.08)&\textbf{7.86}\,(1.24)&\textbf{8.25}\,(1.39)&\textbf{8.58}\,(2.00)\\
\hline
iv)  $\kappa=1$, $\lambda=n^2$  &7.06\,(1.49)&7.78\,(1.65)&8.91\,(1.62)&10.27\,(2.09)&12.39\,(1.97)&13.38\,(2.19)&14.58\,(1.94)&15.45\,(2.13)&17.25\,(2.15)\\
 \hline
v) $\kappa=n$, $\lambda=1$ & 6.78\,(1.03)&8.49\,(0.73)&11.35\,(0.75)&17.01\,(0.82)&24.27\,(0.62)&33.78\,(0.57)&52.92\,(0.56)&74.76\,(0.75)&105.98\,(0.74)\\
 \hline
\end{tabular}}
\end{table}

\begin{table}[tbp]
\centering
\caption{False discovery rate  in Experiment $A_3$}
\label{table:fdr2}
\scriptsize{%
 \setlength{\tabcolsep}{1pt}%
 \renewcommand{\arraystretch}{1.3}%
\begin{tabular}{|l||l|l|l|l|l|l|l|l|l|}
 \hline
Method~\textbackslash~n&  50&100&200&500&1\,000&2\,000&5\,000&10\,000&20\,000 \\
 \hline
 \hline
i) $\kappa=1$, $\lambda=1$ &0.50\,(0.00)&0.56\,(0.10)&0.29\,(0.11)&0.16\,(0.09)&0.15\,(0.05)&0.11\,(0.05)&0.14\,(0.06)&0.11\,(0.06)&0.12\,(0.07)\\
 \hline
ii) $\kappa=1$, $\lambda=\sqrt{n}$ &0.48\,(0.05)&0.25\,(0.07)&0.22\,(0.08)&0.12\,(0.08)&0.03\,(0.06)&0.10\,(0.05)&0.13\,(0.06)&0.11\,(0.06)&0.12\,(0.07)\\
 \hline
iii)  $\kappa=1$, $\lambda=n+1$&0.05\,(0.04) &0.02\,(0.02)&0.03\,(0.04)&0.01\,(0.02)&0.03\,(0.03)&0.02\,(0.03)&0.02\,(0.02)&0.01\,(0.02)&0.01\,(0.03)\\
 \hline
iv)  $\kappa=1$, $\lambda=n^2$  &0.00\,(0.01)&0.00\,(0.00)&0.00\,(0.00)&0.00\,(0.00)&0.00\,(0.00)&0.00\,(0.00)&0.00\,(0.00)&0.00\,(0.00)&0.00\,(0.00)\\
 \hline
v) $\kappa=n$, $\lambda=1$ & 0.5\,(0.00)&0.75\,(0.00)&0.88\,(0.00)&0.95\,(0.00)&0.98\,(0.00)&0.99\,(0.00)&1.00\,(0.00)&1.00\,(0.00)&1.00\,(0.00)\\
 \hline
\end{tabular}}
\end{table}

\begin{table}[tbp]
\centering
\caption{True positive rate  in Experiment $A_3$}
\label{table:tpr2}
\scriptsize{%
 \setlength{\tabcolsep}{1pt}%
 \renewcommand{\arraystretch}{1.3}%
\begin{tabular}{|l||l|l|l|l|l|l|l|l|l|}
 \hline
Method~\textbackslash~n&  50&100&200&500&1\,000&2\,000&5\,000&10\,000&20\,000 \\
 \hline
 \hline
i) $\kappa=1$, $\lambda=1$&1.00\,(0.00) &1.00\,(0.00)&1.00\,(0.01)&1.00\,(0.02)&0.99\,(0.02)&0.98\,(0.03)&0.97\,(0.03)&0.96\,(0.04)&0.94\,(0.05)\\
 \hline
ii) $\kappa=1$, $\lambda=\sqrt{n}$ &1.00\,(0.00)&1.00\,(0.00) &1.00\,(0.01)&0.99\,(0.02)&0.99\,(0.02)&0.97\,(0.03)&0.97\,(0.03)&0.96\,(0.04)&0.94\,(0.05)\\
 \hline
iii)  $\kappa=1$, $\lambda=n+1$&1.00\,(0.01) &1.00\,(0.01)&0.99\,(0.02)&0.98\,(0.03)&0.96\,(0.04)&0.96\,(0.04)&0.94\,(0.05)&0.92\,(0.05)&0.90\,(0.06)\\
 \hline
iv)  $\kappa=1$, $\lambda=n^2$  &0.95\,(0.04)&0.93\,(0.06)&0.88\,(0.06)&0.83\,(0.08)&0.76\,(0.09)&0.73\,(0.10)&0.69\,(0.09)&0.65\,(0.10)&0.59\,(0.10)\\
 \hline
v) $\kappa=n$, $\lambda=1$ &1.00\,(0.00) &1.00\,(0.00)&1.00\,(0.00)&1.00\,(0.00)&1.00\,(0.00)&1.00\,(0.0v0)&1.00\,(0.00)&1.00\,(0.00)&1.00\,(0.00)\\
 \hline
\end{tabular}}
\end{table}

\section{Proofs}\label{sec:proofs}

\subsection{Proof of Theorem \ref{thm:discrete}}
\label{sec:proofdiscrete}

\begin{proof}
  We start the proof by introducing some additional notation: let
  $\Gamma_n$ be the distribution on $\beta$ induced by $\Lambda_n$ and
  the mapping $\beta(\alpha)$, with density
  \[
    \gamma_n(\beta)
      = \frac{\der \Gamma_n(\beta)}{\der \beta}
      = \frac{\der \Lambda_n(\alpha(\beta))}{\der \beta}
      = 2\lambda_n(\alpha(\beta))\sqrt{\alpha(\beta)(1-\alpha(\beta))},
  \]
  for which condition \eqref{cond:prior} implies that
  \begin{equation}\label{cond:phiprior}
    \frac{\sup_{\beta \in [\beta_j,\beta_{j+1}]} \gamma_n(\beta)}{\inf_{\beta
    \in[\beta_j,\beta_{j+1}]} \gamma_n(\beta)}\leq e^{L\sqrt{n}\delta_k}
    \quad \text{for $j=0,\ldots,k$,}
  \end{equation}
  with $\beta_0 = 0$ and $\beta_{k+1} = \pi/2$. Then fix an arbitrary
  $\mlalpha \in [0,1]$, and let $P_\beta(\mlalpha) =
  P_{\alpha(\beta)}(n,\mlalpha)$. Now take $\mlj \in \{0,\ldots,k\}$ such
  that $\beta_{\mlj} \leq \mlphi \leq \beta_{\mlj+1}$ contains the maximum
  likelihood $\beta$-parameter $\mlphi=
  \beta(\mlalpha) = \argmax_\beta P_\beta(\mlalpha) $.
  
  Let us first deal with the second inequality in
  \eqref{eqn:approxBernoulli}, which follows with $C_L = C_1 + C_2 +
  C_1C_2$ by combining the following two assertions:
  \begin{align}
    \int_0^{\pi/2} P_\beta(\mlalpha) \gamma_n(\beta) \der \beta
      &\leq \left(1 + \frac{C_1}{m}\right)
        \int_{[0,\pi/2] \setminus \A}
          P_\beta(\mlalpha) \gamma_n(\beta) \der \beta,
      \label{eqn:noA}\\
    \int_{[0,\pi/2] \setminus \A}
          P_\beta(\mlalpha) \gamma_n(\beta) \der \beta
      &\leq \left(1 + \frac{C_2}{m}\right)
        \sum_{j=1}^k P_{\beta_j}(\mlalpha) \approxLambda_n(\alpha_j).
      \label{eqn:discretize}
  \end{align}
  Here $\A = [\beta_{\mlj},\beta_{\mlj+1}]$, and $C_1 = 4 e^{L\pi/4 +
  \pi^2}$ and $C_2 = L\pi$ are constants. We will also use that $m >
  \max\{C_1,C_2\}$, which is implied by the assumption that $m > C_L$.

  To quantify the approximation error when we change $\beta$ in
  $P_\beta(\mlalpha)$, we will require the following lemma: 
  \begin{lemma}[Lemma~3 of \cite{DeRooijVanErven2009}]\label{lem:discr1}
    Let $\hat \beta = \arg\max_\beta P_\beta(\mlalpha)$ and suppose
    $\beta_1,\beta_2,\hat\beta \in (0,\pi/4]$.
    Then
    \[\ln\frac{P_{\beta_1}(\mlalpha)}{P_{\beta_2}(\mlalpha)}\le
    4n(\beta_2-\beta_1)(\beta_2-\hat \beta)(1 \bmax \frac{\hat \beta}{\beta_2}).\]
  \end{lemma}

  Then, to prove assertions \eqref{eqn:noA} and \eqref{eqn:discretize},
  let
  \[
    \B =
    \begin{cases}
      [\beta_{\mlj} - m \delta_k, \beta_{\mlj+1} + m \delta_k] \intersection
      [m\delta_k/2,\pi/4] & \text{if $\mlphi \leq \pi/4$,}\\
      [\beta_{\mlj} - m \delta_k, \beta_{\mlj+1} + m \delta_k] \intersection
      [\pi/4, \pi/2 - m\delta_k/2] & \text{if $\mlphi > \pi/4$}
    \end{cases}
  \]
  be an interval around $\mlphi$ of width that is roughly proportional
  to $1/\sqrt{n}$, but that does not come too close to the boundary of the
  domain of $\beta$ and also does not cross over the midpoint $\pi/4$. We
  observe that $\B$ is at least $m/2$ times as wide as $\A$. If the prior
  on $\beta$ were uniform, then the prior mass of $\B$ would therefore be
  at least $m/2$ times the prior mass of $\A$.
  Applying~\eqref{cond:phiprior} $m+1$ times (the maximum number of
  intervals between discretization points that $\B$ extends away from
  $\A$), we obtain an approximate version of this statement:
  \begin{equation}\label{eqn:priorratio}
    \frac{\Gamma_n(\A)}{\Gamma_n(\B)}
      \leq \frac{\sup_{\beta \in \A} \gamma_n(\beta)}{\inf_{\beta
      \in \B} \gamma_n(\beta) }\frac{\delta_k}{m\delta_k/2}
      \leq \frac{2 e^{(m+1)L\sqrt{n}\delta_k}}{m}
      \leq \frac{2 e^{L\pi/4}}{m}.
  \end{equation}
  Let us consider the case $\mlphi \leq \pi/4$ (the case $\mlphi >
  \pi/4$ follows by symmetry). Applying \eqref{eqn:priorratio} and
  Lemma~\ref{lem:discr1} with $\beta_1 = \mlphi$ and $\beta_2 = \beta \in
  \B$, we obtain:
  \begin{align}
    \int_{\A} P_\beta(\mlalpha) \gamma_n(\beta) \der \beta
      &\leq \Gamma_n(\A) P_{\mlphi}(\mlalpha)
      = \frac{\Gamma_n(\A)}{\Gamma_n(\B)} \int_{\B} P_{\mlphi}(\mlalpha)
      \gamma_n(\beta) \der \beta \notag\\
      &\leq \frac{\Gamma_n(\A)}{\Gamma_n(\B)} \int_{\B} 
      e^{4n(\beta - \mlphi)^2(1 \bmax \frac{\mlphi}{\beta})}
      P_{\beta}(\mlalpha)
      \gamma_n(\beta) \der \beta \notag\\
      &\leq \frac{\Gamma_n(\A)}{\Gamma_n(\B)} 
      e^{16 n(m+1)^2 \delta_k^2}
      \int_{\B} P_{\beta}(\mlalpha) \gamma_n(\beta) \der \beta \notag\\
      &\leq \frac{2 e^{L\pi/4 + \pi^2}}{m} 
      \int_{\B} P_{\beta}(\mlalpha) \gamma_n(\beta) \der \beta,
      \label{eqn:AlessthanB}
  \end{align}
  from which \eqref{eqn:noA} follows under our assumption that $m > C_1$.

  Next we deal with \eqref{eqn:discretize} and note that, by symmetry,
  it is sufficient to verify
  \[
    \int_0^{\beta_{\mlj}} P_\beta(\mlalpha) \gamma_n(\beta) \der \beta
      \leq (1 + \frac{C_2}{m}) \sum_{j=1}^{\mlj} P_{\beta_j}(\mlalpha)
      \approxLambda_n(\alpha_j).
  \]
  On this interval, which lies left of $\mlphi$, the likelihood
  $P_\beta(\mlalpha)$ is increasing in $\beta$ (as follows e.g.\ from concavity
  of the log-likelihood), so we may upper bound the left-hand side by
  moving prior mass further to the right. By applying assertion
  \eqref{cond:phiprior} twice we can control how closely our prior on
  discretization points approximates a move of probability mass to the
  right: for $j=1,\ldots,k$ we have
  \begin{equation}\label{eqn:moveprior}
    \frac{\Gamma_n([\beta_{j-1},\beta_j])}{\Gamma_n([\beta_{j}-\delta_k/2,\beta_j+\delta_k/2])}
      \leq \frac{\sup_{\beta \in [\beta_{j-1},\beta_{j+1}]}
      \gamma_n(\beta)}{\inf_{\beta \in [\beta_{j-1},\beta_{j+1}]} \gamma_n(\beta)}
      \leq e^{2 L\sqrt{n}\delta_k}\\
      \leq e^{C_2/(2 m)}
      \leq 1+C_2/m,
  \end{equation}
  where we have used that $e^x \leq 1+2x$ for $x \in [0,1/2]$. We
  therefore find that
  \begin{align*}
    \int_0^{\beta_{\mlj}} P_\beta(\mlalpha) \gamma_n(\beta) \der \beta
      &= \sum_{j=1}^{\mlj} \int_{\beta_{j-1}}^{\beta_j} P_\beta(\mlalpha)
      \gamma_n(\beta) \der \beta\\
      &\leq \sum_{j=1}^{\mlj} P_{\beta_j}(\mlalpha)
      \Gamma_n([\beta_{j-1},\beta_j])\\
      &\leq \Big(1+\frac{C_2}{m}\Big)\sum_{j=1}^{\mlj} P_{\beta_j}(\mlalpha)
      \Gamma_n([\beta_{j}-\delta_k/2,\beta_j+\delta_k/2])\\
      &= \Big(1+\frac{C_2}{m}\Big)\sum_{j=1}^{\mlj} P_{\beta_j}(\mlalpha)
      \approxLambda_n(\alpha_j),
  \end{align*}
  as required.

  It remains to prove the first inequality in
  \eqref{eqn:approxBernoulli}, which follows by similar reasoning as
  before, but now from the inequalities
  \begin{align}
    \int_0^{\pi/2} P_\beta(\mlalpha) \gamma_n(\beta) \der \beta
      &\geq \left(1 - \frac{C_1}{2 m}\right)
        \left(\int_0^{\pi/2} + \int_{\A}\right)
          P_\beta(\mlalpha) \gamma_n(\beta) \der \beta,
        \label{eqn:addA}\\
      \left(\int_0^{\pi/2} + \int_{\A}\right)
          P_\beta(\mlalpha) \gamma_n(\beta) \der \beta
      &\geq \left(1 - \frac{C_3}{m}\right)
        \sum_{j=1}^k P_{\beta_j}(\mlalpha) \approxLambda_n(\alpha_j),
      \label{eqn:reversediscretize}
  \end{align}
  where $C_1$ is the same constant as above, $C_3 = \pi^2/4 + C_2$, and
  we now only need that $C_L \geq C_1/2 + C_3$, which is satisfied by
  our previous choice.

  To prove \eqref{eqn:addA}, we note that it readily follows from
  \eqref{eqn:AlessthanB}, so it remains only to
  establish~\eqref{eqn:reversediscretize}. To this end, we need the
  following inverse version of \eqref{eqn:moveprior}:
  \[
    \frac{\Gamma_n([\beta_j,\beta_{j+1}])}{\Gamma_n([\beta_{j}-\delta_k/2,\beta_j+\delta_k/2])}
      \geq \frac{\inf_{\beta \in [\beta_{j-1},\beta_{j+1}]}
      \gamma_n(\beta)}{\sup_{\beta \in [\beta_{j-1},\beta_{j+1}]} \gamma_n(\beta)}
      \geq \frac{1}{1+C_2/m} \geq 1 - \frac{C_2}{m}.
  \]
  Then, again using that the likelihood $P_\beta(\mlalpha)$ is increasing in
  $\beta$ on the left of $\mlphi$, we see that:
  \begin{align}
      \int_0^{\beta_{\mlj}} P_\beta(\mlalpha) \gamma_n(\beta) \der \beta
        &\geq \sum_{j=1}^{\mlj-1} \int_{\beta_j}^{\beta_{j+1}} P_\beta(\mlalpha)
        \gamma_n(\beta) \der \beta\notag\\
        &\geq \sum_{j=1}^{\mlj-1}  P_{\beta_j}(\mlalpha)
        \Gamma_n([\beta_j,\beta_{j+1}])\notag\\
        &\geq (1-\frac{C_2}{m}) \sum_{j=1}^{\mlj-1}  P_{\beta_j}(\mlalpha)
        \approxLambda_n(\alpha_j),
        \label{eqn:reverseLeft}
  \end{align}
  and, by symmetry,
  \begin{equation}
    \label{eqn:reverseRight}
    \int_{\beta_{\mlj+1}}^{\pi/2} P_\beta(\mlalpha) \gamma_n(\beta) \der \beta
      \geq (1-\frac{C_2}{m}) \sum_{j=\mlj+2}^{k}  P_{\beta_j}(\mlalpha)
        \approxLambda_n(\alpha_j).
  \end{equation}
  If $\mlj=0$ or $\mlj=k$, then one of the last two inequalities implies
  \eqref{eqn:reversediscretize} and we are done. Otherwise, $\mlj \in
  \{1,\ldots,k-1\}$ and by applying Lemma~\ref{lem:discr1} with $\beta_1
  = \mlphi$ and $\beta_2=\beta \in \A$ we get
  \begin{align}
    \int_{\beta_{\mlj}}^{\beta_{\mlj+1}} P_\beta(\mlalpha) &\gamma_n(\beta) \der \beta
      \geq e^{-4n\delta_k^2 (1 \bmax \frac{\mlphi}{\beta})} P_{\mlphi}(\mlalpha)
      \Gamma_n([\beta_{\mlj},\beta_{\mlj+1}])\notag\\
      &\geq e^{-3\pi^2/(4(m+1)^2)} P_{\mlphi}(\mlalpha)
      \Gamma_n([\beta_\mlj,\beta_{\mlj+1}])\notag \\
      &\geq e^{-3\pi^2/(4(m+1)^2)} \Big(1 - \frac{C_2}{m}\Big)\max\big\{P_{\beta_{\mlj}}(\mlalpha)
      \approxLambda_n(\alpha_{\mlj}),P_{\beta_{\mlj+1}}(\mlalpha)
      \approxLambda_n(\alpha_{\mlj+1})\big\}\notag\\
      &\geq \Big(1-\frac{\pi^2/4 + C_2}{m}\Big)\max\big\{P_{\beta_{\mlj}}(\mlalpha)
      \approxLambda_n(\alpha_{\mlj}),P_{\beta_{\mlj+1}}(\mlalpha)
      \approxLambda_n(\alpha_{\mlj+1})\big\},
    \label{eqn:reverseA}
  \end{align}
  where we have used that $m > C_L \geq 2$ in the last inequality.
  Adding up \eqref{eqn:reverseLeft}, \eqref{eqn:reverseRight}, and twice \eqref{eqn:reverseA}, we obtain \eqref{eqn:reversediscretize}, completing the proof of the theorem.
\end{proof}

\subsection{Proof of Theorem \ref{thm:equivalence:prior}}\label{sec:equivalence:prior}

We note that the model selection prior can be represented in spike-and-slab form \eqref{prior:spike-and-slab}  if and only if 
\begin{equation}\label{eqn:definetti}
{n \choose s}^{-1}\pi_n(s)=\int_0^1 \alpha^s (1-\alpha)^{n-s}d\Lambda_n(\alpha),\qquad \text{for all $s=0,1,...,n$}.
\end{equation}
This is closely related to a finite-sample version of de Finetti's
theorem for Bernoulli sequences: on the left-hand side of
\eqref{eqn:definetti} we have an exchangeable distribution on binary
sequences of length $n$ with $s$ ones, and on the right-hand side we
want to find the corresponding mixture $\Lambda_n$ of independent,
identically distributed Bernoulli random variables. Existing ways to
extend de Finetti's theorem to finite samples include allowing signed
mixtures \citep{KernsSzekely2006} or characterizing how well the
right-hand side can approximate the left-hand side in variational
distance \citep{DiaconisFreedman1980}. However, our setup does not allow
weakening the identity \eqref{eqn:definetti} in any way, so instead we
take the alternative approach of posing necessary and sufficient
conditions on $\pi_n$ such that \eqref{eqn:definetti} holds exactly.

Let us decompose the probability measure $\Lambda_n(\alpha)$ as a sum of a point mass at $\alpha=1$ and a measure $\tilde{\Lambda}_n$ which puts zero mass at $\alpha=1$, i.e. $\tilde{\Lambda}_n(\alpha)=\Lambda_n(\alpha)-\Lambda_n(1)\delta_{1}$. Then \eqref{eqn:definetti} can be written in the form
\begin{equation*}
\begin{split}
{n \choose s}^{-1}\pi_n(s)&=\int_0^1 \alpha^s (1-\alpha)^{n-s}d\tilde{\Lambda}_n(\alpha),\qquad \text{for all $s=0,1,...,n-1$},\\
\pi_n(n)-\Lambda_n(1)&=\int_0^1 \alpha^n d\tilde{\Lambda}_n(\alpha).
\end{split}
\end{equation*}
Next let us substitute $\alpha=u/(1+u)$ in the right-hand side of the
preceding displays, which makes them equal to
\begin{align*}
\int_0^1 \Big(\frac{\alpha}{1-\alpha}\Big)^{s}(1-\alpha)^n d\tilde{\Lambda}_n(\alpha)
  =\int_0^\infty u^s\frac{1}{(1+u)^{n}}d\tilde{\Lambda}_n\big(\frac{u}{1+u}\big),\qquad s=0,1,...,n.
\end{align*}
Note that since $\Lambda_n(1) \in [0,\pi_n(n)]$ can be chosen
arbitrarily, the parameter $c_n=\pi_n(n)-\Lambda_n(1)\in[0,\pi_n(n) ]$
can take any arbitrary value. Then by denoting the measure
$(1+u)^{-n}d\tilde{\Lambda}_n\big(\frac{u}{1+u}\big)$ on $[0,\infty)$ by
$d\bar\Lambda_n(u)$ we arrive at the equations
\begin{align*}
&\int_0^{\infty}u^{s}d\bar\Lambda_n(u)={n \choose s}^{-1}\pi_n(s),\qquad \text{for all $s=0,1,...,n-1$},\\
&\int_0^{\infty}u^{n}d\bar\Lambda_n(u)=c_n.
\end{align*}
This is called the truncated (or finite/reduced) Stieltjes moment
problem and the sufficient and necessary conditions for the existence  of a 
general Radon measure $\bar\Lambda_n$ on $[0,\infty)$, 
that satisfies the above equation system coincide with the conditions of
our theorem. See, for instance, Theorems~9.35 and~9.36 of
\cite{schmudgen:2017} for the odd and even case, respectively.

We note that all steps above are reversible: if, in view of the
truncated  Stieltjes moment problem, a measure $\bar{\Lambda}_n$ exists
for some $c_n\in[0,\pi_n(n)]$, then one can construct the measure
$\Lambda_n(\alpha)=(1-\alpha)^{-n}d\bar{\Lambda}_n\big(\alpha/(1-\alpha)\big)+(\pi_n(n)-c_n)\delta_{1}$
satisfying \eqref{eqn:definetti}. One can also see that $\Lambda_n$ will
then be a probability measure using Fubini's theorem:
\begin{align*}
1&=\sum_{s=0}^n\pi_n(s)= \sum_{s=0}^n\int_0^1 {n \choose s}\alpha^s (1-\alpha)^{n-s}d\Lambda_n(\alpha)\\
&= \int_0^1 \sum_{s=0}^n{n \choose s}\alpha^s (1-\alpha)^{n-s}d\Lambda_n(\alpha)=\Lambda_n([0,1]).
\end{align*}

\subsection{Proofs for the Examples from
Section~\ref{sec:connection}}\label{sec:examples}

\subsubsection{Proof of Example \ref{example:1}}
Let us take $c_n=p^n$ if $p\in[0,1)$, and $c_n=0$ if $p=1$.  Then the
vector $\mu$ in Theorem \ref{thm:equivalence:prior} takes the form
$\mu=\big(p^s(1-p)^{n-s}\big)_{s=0,1,...,n}$ for $p\in[0,1)$ and
$\mu=(0,\ldots,0)$ of length $n+1$ if $p=1$.  

Let us consider first the odd case $n=2k+1$. For $p=1$, both Hankel
matrices are the zero matrix, which is positive semi-definite, and the
zero-vector $\mu$ is inside of the column space of the first matrix.
Next assume that $p<1$. Then the first Hankel matrix $H_k(\mu)$ is
positive semi-definite (its eigenvalues are $\lambda_1 = (1-p)\sum_{\ell=0}^k
p^{2(k-\ell)}(1-p)^{2\ell}>0$ and $\lambda_2 = \cdots = \lambda_{k+1} =
0$). Similarly,
the second Hankel matrix $H_k(F\mu)$ is also positive semi-definite
(its eigenvalues are $\lambda_1 = p\sum_{\ell=0}^k p^{2(k-\ell)}(1-p)^{2\ell}\geq
0$ and $\lambda_2 = \cdots = \lambda_{k+1} = 0$). Finally, note that the vector
$v=(p^{k+\ell}(1-p)^{k+1-\ell})_{\ell=1,...,k+1}^\top$ is inside of the
column space of  $H_k(\mu)$ since $v$
is equal to $p/(1-p)$ times the last column of the matrix.

The even case $n=2k$ follows by similar arguments.

\subsubsection{Proof of Example \ref{example:2}}
Let us take $c_n=\pi_n(n) \propto \lambda^n e^{-\lambda}/n!$. Then the
vector $\mu$ in Theorem \ref{thm:equivalence:prior} takes the form $\mu
\propto \big(\lambda^s e^{-\lambda}(n-s)!/n! \big)_{s=0,1,...,n}$. 

Then let us consider first the odd case $n=2k+1$. We show that the
determinants of the leading principal minors of the Hankel matrices
$H_k(\mu)$ and $H_k(F\mu)$ are both positive for every $\ell \leq k$,
which implies that both matrices are positive definite. First we note
that by multiplying the rows by a positive constant the sign of the
determinant remains unchanged; therefore the determinant of the leading
principal minor of $H_k(\mu)$ of order $\ell+1$ has the same sign as the following matrix
\begin{align*}
  \begin{bmatrix}
    n! & (n-1)!\,\lambda & (n-2)!\,\lambda^2 & ... & (n-\ell)!\,\lambda^{\ell} \\
   (n-1)!\,\lambda & (n-2)!\,\lambda^2 & (n-3)!\,\lambda^3& ... &
   (n-\ell-1)!\,\lambda^{\ell+1}\\
   \vdots & \vdots& \vdots &\ddots&  \vdots\\
 (n-\ell)!\,\lambda^{\ell} & (n-\ell-1)!\,\lambda^{\ell+1} &
 (n-\ell-2)!\,\lambda^{\ell+2} &...&(n-2\ell)!\,\lambda^{2\ell}
  \end{bmatrix}.
\end{align*}
Then for computational convenience we note that the determinant of the
matrix does not change by mirroring it in the central point, i.e.\ transforming the matrix $A=(a_{i,j})_{1\leq i,j\leq n}$ into $B=(a_{n+1-i,n+1-j})_{1\leq i,j\leq n}$. Hence the preceding matrix has the same determinant as 
\begin{align}\label{mtx: modif}
  \begin{bmatrix}
  (n-2\ell)!\,\lambda^{2\ell} & ... &(n-\ell-2)!\,\lambda^{\ell+2}  & (n-\ell-1)!\,\lambda^{\ell+1} & (n-\ell)!\,\lambda^{\ell}  \\
   \vdots &\ddots & \vdots& \vdots&  \vdots\\
  (n-\ell-1)!\,\lambda^{\ell+1} & ... & (n-3)!\,\lambda^3 &  (n-2)!\,\lambda^2& (n-1)!\,\lambda \\
(n-\ell)!\,\lambda^{\ell}  & ...  & (n-2)!\,\lambda^2 &(n-1)!\,\lambda &  n! 
  \end{bmatrix}.
\end{align}
We also note that subtracting a multiple of a row from another does not
change the determinant of the matrix. Using this elementary step we will
derive an upper triangular matrix from the preceding one with positive
elements in the diagonal, which implies that the matrix has positive determinant. In the following we will use iteratively that 
$$(n-s_1)!-(n-s_1-1)!\,(n-s_2)=(n-s_1-1)!\, (s_2-s_1).$$
Then by subtracting $(n-\ell)/\lambda$ times the one before the last row
from the last row in \eqref{mtx: modif}, then $(n-\ell-1)/\lambda$ times
the two before the last row from the one before the last row and so on,
finishing with subtracting $(n-2\ell+1)/\lambda$ times the first row
from the second row, we get the matrix
\begin{align*}
  \begin{bmatrix}
  (n-2\ell)!\,\lambda^{2k} &  (n-2\ell+1)!\,\lambda^{2\ell-1} &(n-2\ell+2)!\,\lambda^{2\ell-2}  &... & (n-\ell)!\,\lambda^{\ell}  \\
0&(n-2\ell+1)!\,\lambda^{2\ell-2}   & (n-2\ell+2)!\,  2\lambda^{2\ell-3} & ...& (n-\ell)!\, \ell\lambda^{\ell-1}  \\
   \vdots &\vdots& \vdots & \ddots&  \vdots\\
 0& (n-\ell-1)!\,\lambda^{\ell}  & (n-\ell)!\,2\lambda^{\ell-1}   & ... & (n-2)!\,\ell\lambda \\
0&(n-\ell)!\,\lambda^{\ell-1} &(n-\ell+1)!\,2\lambda^{\ell-2}   & ...  &  (n-1)!\,\ell 
  \end{bmatrix}.
\end{align*}
As a next step we subtract again $(n-\ell)/\lambda$ times the one before the last row from the last row in \eqref{mtx: modif}, then $(n-\ell-1)/\lambda$ times the two before the last row from the one before the last row and so on finishing with  subtracting $(n-2\ell+2)/\lambda$ times the second row from the third row we get the matrix

\begin{align*}
  \begin{bmatrix}
  (n-2\ell)!\,\lambda^{2\ell} &  (n-2\ell+1)!\,\lambda^{2\ell-1} &(n-2\ell+2)!\,\lambda^{2\ell-2}  &... & (n-\ell)!\,\lambda^{\ell}  \\
0&(n-2\ell+1)!\,\lambda^{2\ell-2}   & (n-2\ell+2)!\,  2\lambda^{2\ell-3} & ...& (n-\ell)!\, \ell\lambda^{\ell-1}  \\
0&0  & (n-2\ell+2)!\,  2\lambda^{2\ell-4} & ...& (n-\ell)!\, \ell(\ell-1)\lambda^{\ell-2}  \\
   \vdots &\vdots& \vdots & \ddots&  \vdots\\
 0&0  & (n-\ell+1)!\,2\lambda^{\ell-1}   & ... & (n-2)!\,\ell(\ell-1)\lambda \\
0&0 &(n-\ell)!\,2\lambda^{\ell-2}   & ...  &  (n-1)!\,\ell(\ell-1)
  \end{bmatrix}.
\end{align*}
By iterating this algorithm we get an upper triangular matrix which has
positive values in the diagonal, finishing the proof of our statement.
The positive definiteness of the second Hankel matrix $H_k(F\mu)$
follows similarly. Finally note that since $H_k(\mu)$ is positive
definite every $(k+1)$-dimensional vector is inside of its column space,
including $v=(\mu_{k+1},\ldots,\mu_n)$.

The even case $n=2k$ follows by similar arguments.

\subsubsection{Proof of Example \ref{example:3}}
Let us consider the determinant of the leading principal minor of
$H_{k}(\mu)$ of order 2, where $n=2k$ or $n=2k+1$. The determinant of
this matrix is proportional to
\begin{align*}
  \det\begin{bmatrix}
 1& \frac{1}{n}\\
\frac{1}{n} & \frac{2^{-\lambda}}{n(n-1)/2}
  \end{bmatrix}= \frac{2^{1-\lambda}}{n(n-1)}-\frac{1}{n^2},
\end{align*}
which is negative for $n> 2^{\lambda-1}/(2^{\lambda-1}-1)$.
Hence the conditions of Theorem \ref{thm:equivalence:prior} do not hold and therefore the prior cannot be written in spike-and-slab form.

\subsubsection{Proof of Example \ref{example:4}}
Let us consider the determinant of the leading principal minor of
$H_{k}(\mu)$ of order 2, where $n=2k$ or $n=2k+1$. The determinant of
this matrix is proportional to
\begin{align*}
  \det\begin{bmatrix}
 1& \frac{e^{-1}}{n}\\
\frac{e^{-1}}{n} & \frac{e^{-2^\lambda}}{n(n-1)/2}
  \end{bmatrix}= \frac{2e^{-2^\lambda}}{n(n-1)}-\frac{e^{-2}}{n^2},
\end{align*}
which is negative for $n>c/(c-1)$ with $c=e^{2^{\lambda}-2}/2>1$. Hence
the conditions  of Theorem \ref{thm:equivalence:prior} do not hold and
the prior cannot be written in spike-and-slab form.
\end{document}